%% file: main.tex
\documentclass[11pt]{article}
\include{macros}

\usepackage{authblk}
\title{LDPC Codes Achieve List Decoding Capacity\thanks{
JM is partially supported by NSF grants CCF-1814603 and CCF-1563742. A significant portion of this work was accomplished while JM was a postdoctoral fellow at the Weizmann Institute, partially supported by Irit Dinur's ERC-CoG grant 772839. NRe is partially supported by  NSERC grant CGSD2-502898, NSF grants CCF-1422045, CCF-1814603, CCF-1527110, CCF-1618280, CCF-1910588, NSF CAREER award CCF-1750808 and a Sloan Research Fellowship. NRo Is partially supported by BSF grant 2014359 and ISF grant 735/20.
SS and  MW are partially supported by NSF grants CCF-1844628, CCF-1814629, and a Sloan Research Fellowship.  SS is partially supported by a Google Graduate Fellowship. 
} 
}
\author[1]{Jonathan Mosheiff}
\author[1]{Nicolas Resch}
\author[2]{Noga Ron-Zewi} 
\author[3]{Shashwat Silas}
\author[3]{Mary Wootters}
\affil[1]{Carnegie Mellon University}
\affil[2]{University of Haifa}
\affil[3]{Stanford University}
\begin{document}

\maketitle

\begin{abstract}
We show that Gallager's ensemble of Low-Density Parity Check (LDPC) codes achieves list-decoding capacity with high probability. These are the first graph-based codes shown to have this property.  
This result opens up a potential avenue towards truly linear-time list-decodable codes that achieve list-decoding capacity.

Our result on list decoding follows from a much more general result: any \em local \em property satisfied with high probability by a random linear code is also satisfied with high probability by a random LDPC code from Gallager's distribution.  Local properties are properties characterized by the exclusion of small sets of codewords, and include list-decodability, list-recoverability and average-radius list-decodability. 

In order to prove our results on LDPC codes, we establish sharp thresholds for when local properties are satisfied by a random linear code.  More precisely, we show that for any local property $\mathcal{P}$, there is some $R^*$ so that random linear codes of rate slightly less than $R^*$ satisfy $\mathcal{P}$ with high probability, while random linear codes of rate slightly more than $R^*$, with high probability, do not.  We also give a characterization of the threshold rate $R^*$.
\end{abstract}

\thispagestyle{empty}
\setcounter{page}{0}
\newpage
\section{Introduction}\label{sec:intro}
\input{intro}

\section{High-level idea: proof of Theorem~\ref{thm:main}}\label{sec:mainthm}
\input{high_level}

\section{Sharp thresholds of local properties for random linear codes: proof of Lemma~\ref{lem:tau-threshold}}\label{sec:random_linear_matrix_char}
\input{RandomLinearMatrix}

\section{Matrices contained in a random LDPC code: proof of Lemma~\ref{lem:ProbabilityOfMatrixInLDPC}}\label{sec:random_ldpc_matrix}
\input{RandomLDPCMatrix}

\section{Random LDPC codes achieve the GV bound: proof of Theorem~\ref{thm:ldpc_distance}}\label{sec:distance}
\input{distance}

\section*{Acknowledgements} The first author would like to thank Yael Hacohen and Nati Linial for useful conversations. The second author would like to thank Venkat Guruswami for helpful feedback on a draft of this work.  We thank anonymous reviewers for helpful comments.

\bibliographystyle{alpha}
\bibliography{refs}

\end{document}

%% file: macros.tex
\usepackage[margin=1.125in]{geometry}
\usepackage[dvipsnames]{xcolor}
\definecolor{DarkGreen}{rgb}{0.1,0.5,0.1}
\definecolor{DarkRed}{rgb}{0.5,0.1,0.1}
\definecolor{DarkBlue}{rgb}{0.1,0.1,0.5}

\usepackage{mleftright}

\usepackage[small]{caption}
\usepackage[pdftex]{hyperref}
\hypersetup{
    unicode=false,          
    pdftoolbar=true,        
    pdfmenubar=true,        
    pdffitwindow=false,      
    pdfnewwindow=true,      
    colorlinks=true,       
    linkcolor=DarkBlue,          
    citecolor=DarkGreen,        
    filecolor=DarkGreen,      
    urlcolor=DarkBlue,          
    pdftitle={},
    pdfauthor={},    
}
\usepackage{amssymb,amsmath,amsthm,amsfonts,enumitem}
\usepackage{fullpage,nicefrac,comment}
\usepackage{tikz}
\usetikzlibrary{arrows,decorations.pathmorphing,decorations.shapes,snakes,patterns,positioning}
\usepackage[ruled,linesnumbered]{algorithm2e}
\usepackage{thm-restate}

\newcommand{\deffont}{\textsf} 

\newcommand{\B}{\ensuremath{\mathcal{B}}}

\newcommand{\cC}{\ensuremath{{C}}} 
\newcommand{\C}{\ensuremath{{C}}}  
\newcommand{\cM}{\ensuremath{\mathcal{M}}}

\newcommand{\cP}{\ensuremath{\mathcal{P}}}

\newcommand{\cD}{\ensuremath{\mathcal{D}}}
\newcommand{\cI}{\ensuremath{\mathcal{I}}}

\newcommand{\cL}{\ensuremath{\mathcal{L}}}

\newcommand{\R}{{\mathbb R}}

\newcommand{\F}{{\mathbb F}}

\newcommand{\CC}{{\mathbb C}}
\newcommand{\V}{\mathbb{V}}
\newcommand{\N}{\ensuremath{\mathbb{N}}}

\newcommand{\PR}[1]{\mathrm{Pr}\left[ #1\right]}
\newcommand{\E}[1]{\mathbb{E}\left[ #1\right]}
\newcommand{\EE}{\mathbb{E}}
\newcommand{\Eover}[2]{\mathbb{E}_{#1}\left[ #2 \right]}
\newcommand{\Var}[1]{\mathrm{Var}\left( #1\right)}

\newcommand{\inabs}[1]{\left|#1\right|}
\newcommand{\ip}[2]{\ensuremath{\left\langle #1,#2\right\rangle}}

\newcommand{\inset}[1]{\left\{#1\right\}}

\newcommand{\inparen}[1]{\left(#1\right)}

\newcommand{\suchthat}{\,:\,}

\newcommand{\supp}{\mathrm{supp}}

\newcommand{\poly}{\mathrm{poly}}

\newcommand{\rank}{\ensuremath{\operatorname{rank}}}
\newcommand{\spn}{\ensuremath{\operatorname{span}}}
\newcommand{\dist}{\ensuremath{\operatorname{dist}}}

\newcommand{\tr}{\mathrm{tr}}

\newcommand{\DKL}[3]{{D_{\mathrm{KL}}}_{#3}\left(#1\parallel #2\right)}

\newcommand{\divides}{\mbox{ {\Large $|$} } }

\newcommand{\ind}[1]{\ensuremath{\mathbf{1}_{#1}}}

\newcommand{\TRLC}{R_{\mathrm{RLC}}}
\newcommand{\TRLCEXP}{R^{\EE}_{\mathrm{RLC}}}

\newcommand{\CRLC}{C_{\mathrm{RLC}}}
\newcommand{\CLDPC}[1]{C_{{#1}\mathrm{LDPC}}}

\newcommand\numberthis{\addtocounter{equation}{1}\tag{\theequation}}

\newcommand{\eps}{\varepsilon}
\renewcommand{\epsilon}{\varepsilon}

\newcommand{\wt}[1]{\mathrm{wt}\mleft(#1\mright)}

\renewcommand{\exp}[1]{\mathrm{exp}\mleft(#1\mright)}

\theoremstyle{plain}
\declaretheorem[name=Theorem,numberwithin=section]{theorem}
\declaretheorem[name=Lemma,sibling=theorem]{lemma}
\newtheorem*{lemma*}{Lemma} 
\newtheorem*{theorem*}{Theorem} 
\newtheorem{definition}[theorem]{Definition}

\newtheorem{observation}[theorem]{Observation}

\newtheorem{corollary}[theorem]{Corollary} 

\newtheorem{remark}[theorem]{Remark}
\newtheorem{claim}[theorem]{Claim}

\newtheorem{fact}[theorem]{Fact}

\newtheorem{question}[theorem]{Question}
\newtheorem{example}[theorem]{Example}

%% file: intro.tex
In this paper, we study sets $\cC \subset \Sigma^n$ of strings of length $n$, with the combinatorial property that not too many elements of $\cC$ are contained in any small enough Hamming ball.  In the language of coding theory, such a $\cC$ is a \em list-decodable code. \em  List-decoding is an important primitive in coding theory, with applications ranging from communication to complexity theory.
However, as discussed below, most constructions of \em capacity-achieving \em (aka, optimal) list-decodable codes are fundamentally algebraic, despite a rich history of combinatorial---and in particular, graph-based---constructions of error correcting codes.

We show that a random ensemble of \em Low-Density Parity-Check (LDPC) codes \em achieves list-decoding capacity with high probability.
LDPC codes are the prototypical example of graph-based codes, and are popular both in theory and in practice because of their extremely efficient algorithms.  One of the motivations for this work is that we do not currently know any linear-time algorithms for list-decoding any code up to capacity; since graph-based codes offer linear-time algorithms for a variety of other coding-theoretic tasks, our result opens up the possibility of using these constructions for linear-time list-decoding algorithms.

\paragraph{List Decoding.}  Formally, a code $\cC \subset \Sigma^n$ is \deffont{$(\alpha, L)$-list-decodable} if for all $z \in \Sigma^n$, 
\[ |\inset{ c \in \cC \suchthat \dist(c,z) \leq \alpha }| \leq L. \]
Above, $\dist(c,z)$ is the \deffont{relative Hamming distance}, 
\[ \dist(c,z) = \frac{1}{n} |\inset{ i \suchthat c_i \neq z_i }|. \]
Elements $c \in \cC$ are called \deffont{codewords}, $\Sigma$ is called the \deffont{alphabet}, and $n$ is called the \deffont{length} of the code.   

The fundamental trade-off in list-decoding is between the parameter $\alpha$ and the size $|\cC|$ of the code, given that the \deffont{list size} $L$ is reasonably small.
We would like both $\alpha$ and $|\cC|$ to be large, but these requirements are at odds:
the larger the code $\cC$ is, the closer together the codewords have to be, which means that $\alpha$ cannot be as large before some Hamming ball of radius $\alpha$ has many codewords in it.
The size of a code $\cC$ is traditionally quantified by the \deffont{rate} $R$ of $\cC$, which is defined as
$$ R = \frac{ \log_{|\Sigma|}(|\cC|)}{n}. $$
The rate of $\cC$ is a number between $0$ and $1$, and larger rates are better.

List-decoding has been studied since the work of Elias and Wozencraft in the 1950's~\cite{Elias57,Wozencraft58}, and by now we have a good understanding of what is possible and what is not. The classical \emph{list-decoding capacity theorem} states that there exist codes over alphabets of size $|\Sigma| = q$ and of rate $R \geq 1 - h_q(\alpha) - \eps$ which are $(\alpha, 1/\eps)$-list-decodable, where
\begin{equation}\label{eq:qary_entropy}
 h_q(x) := x \log_q(q-1) - x \log_q(x) - (1 - x) \log_q(1 - x) 
\end{equation}
is the $q$-ary entropy function.  Conversely, any such code with rate $R \geq 1 - h_q(\alpha) + \eps$ must have exponential list sizes, in the sense that there is some $z \in \Sigma^n$ so that $|\inset{ c \in \cC \suchthat \dist(c,z) \leq \alpha}| = \mathrm{exp}_{\eps, \alpha}(n)$.\footnote{Here and throughout the paper, $\exp{n}$ denotes $2^{\Theta(n)}$, and subscripts indicate that we are suppressing the dependence on those parameters.}

A code of rate $R \geq 1 - h_q(\alpha) - \eps$ that is $(\alpha, L)$-list decodable for $L = O_{\eps,\alpha}(1)$ is said to \deffont{achieve list-decoding capacity}, and a major question in list-decoding is which codes have this property.  By now we have three classes of examples.
First, it is not hard to see that completely random codes achieve list-decoding capacity with high probability.
Second, a long line of work (discussed more below) has established that \em random linear codes \em do as well: 
we say that a code over the alphabet $\Sigma = \F_q$ is linear if it is a linear subspace of $\F_q^n$,\footnote{Here and throughout the paper, $\F_q$ denotes the finite field with $q$ elements.} and a random linear code is a random subspace.
Third, there are several explicit constructions of codes which achieve list-decoding capacity; as discussed below, most of these constructions rely importantly on algebraic techniques.

\paragraph{LDPC Codes.}
Graph-based codes, such as LDPC codes, are a
class of codes which is notably absent from the list of capacity-achieving codes above.
Originally introduced by Gallager in the 1960's~\cite{Gal62}, codes defined from graphs have become a class of central importance in the past 30 years.

Here is one way to define a code using a graph.
Suppose that $G = (V, W ,E)$ is a bipartite graph with $|V| = n$ and $|W| = m$ for $m \leq n$. Then $G$ naturally defines a linear code $\cC \subset \F_q^n$ of rate at least $1 - m/n$ as follows:
\[ C = \inset{ c \in \F_q^n \suchthat \forall j \in W, \sum_{i \in \Gamma(j)} \alpha_{i,j} c_i = 0 }, \]
where $\Gamma(i)$ denotes the neighbors of $i$ in $G$ and $\alpha_{i,j} \in \F_q$ are fixed coefficients.  (See Figure~\ref{fig:graph_pic}). That is, each vertex in $W$ serves as a \deffont{parity check}, and the code is defined as all possible labelings of vertices in $V$ which obey all of the parity checks. When the right-degree\footnote{That is, the maximum degree of a parity-check node.} of $G$ is small, the resulting code is called a Low-Density Parity Check (LDPC) code.

LDPC codes and related constructions (in particular, Tanner codes~\cite{Tanner81} and expander codes~\cite{SipserS94,Zemor01}) are notable for their efficient algorithms for unique decoding; in fact, the only linear-time encoding/decoding algorithms we have for unique decoding (that is, list-decoding with $L=1$) are based on such codes. 

\paragraph{Motivating question.}
We currently do not know of any linear-time algorithms to list-decode any code to capacity.  Since graph-based codes and LDPC codes in particular are notable for their linear-time algorithms, this state of affairs motivates the following question:

\begin{question}\label{q:main}
Are there (families) of LDPC codes that achieve list-decoding capacity?
\end{question}

\subsection{Contributions}
Motivated by Question~\ref{q:main}, our contributions are as follows.

\begin{itemize}
\item[(1)] We show that the answer to Question~\ref{q:main} is ``yes.'' More precisely, we show that random LDPC codes 
 (the same ensemble studied by Gallager in his seminal work nearly 60 years ago~\cite{Gal62}), 
achieve list-decoding capacity with high probability.
\item[(2)] In fact, we show a stronger result: random LDPC codes
satisfy, with high probability, any 
\em local \em 
property that random linear codes satisfy with high probability.  
We define local properties precisely below; informally, a local property is one defined by the exclusion of certain bad sets.
List-decodability is a local property---it can be defined by the exclusion of any big set of vectors that are too close together---and this answers Question~\ref{q:main}. 
\item[(3)]  Along the way, we develop a characterization of the local properties that are satisfied with high probability by a random linear code.  We show that for any local property $\mathcal{P}$, there is a threshold $R^*$ so that random linear codes of rate slightly less than $R^*$ satisfy $\mathcal{P}$ with high probability, while random linear codes of rate slightly greater than $R^*$ with high probability do not.  Moreover, we give a characterization of the threshold $R^*$.  

In \cite{GLMRSW20}, the above characterization is used  to compute lower bounds on the list-decoding and list-recovery parameters of random linear codes. This additional application does not directly relate to LDPC codes.
\end{itemize}
We describe each of these contributions in more detail below.

\paragraph{(1) Random LDPC codes achieve list-decoding capacity.}
We study the so-called ``Gallager ensemble'' of binary LDPC codes introduced by Gallager in the 1960's~\cite{Gal62}, as well as its natural generalization to larger alphabets.\footnote{For binary codes, our definition coincides with Gallager's. For larger alphabets our definition is somewhat different: Gallager's ensemble chooses the coefficients $\alpha_{i,j}$ to be all ones, while we choose them to be random elements of $\F_q^*$.} 

Fix a rate $R \in (0,1)$ and an integer $s$, and let $t = (1-R)s$.    We assume that $t$ is an integer.
To define the ensemble of \deffont{random $s$-LDPC codes of rate $R$}, we need to specify a distribution on the underlying bipartite graphs and a distribution on the coefficients $\alpha_{i,j}$.  
We define the distribution on graphs as follows.
Let $G_i = (V, W_i, E_i)$ for $i=1,\ldots, t$ be independent uniformly random $(1,s)$-regular\footnote{A $(c,d)$-regular bipartite graph $G$ is a bipartite graph where every vertex in the left partition has degree $c$ and every vertex in the right partition has degree $d$.} bipartite graphs with a shared left vertex set $V$ of size $n$ and disjoint right vertex sets $W_i$, each of size $n/s$.  Then let $G = (V,W,E)$ be the union of these graphs, where $W = \bigcup_{i=1}^t W_i$.  
Finally, we choose the coefficients $\alpha_{i,j}$ for $(i,j) \in E$ to be uniformly random in $\F_q^*$. We refer to $s$ as the \deffont{sparsity parameter}.
The ensemble of random $s$-LDPC codes of rate $R$ is illustrated in Figure~\ref{fig:graph_pic}. 

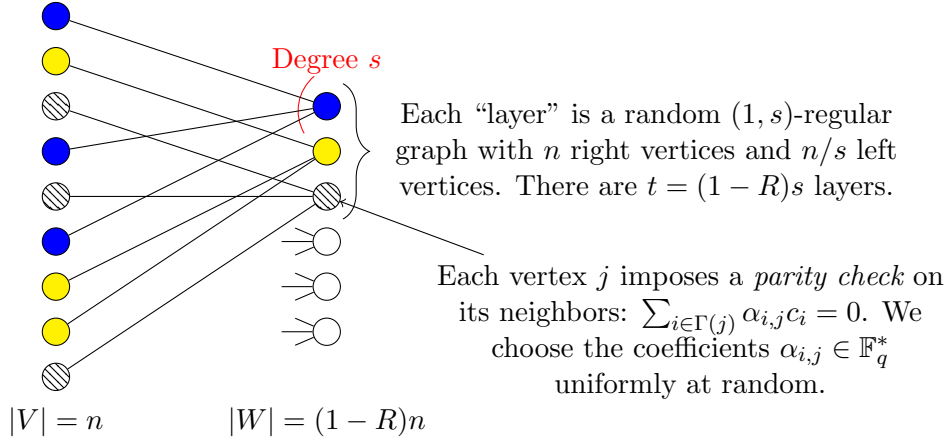
\begin{figure}
\begin{center}
\begin{tikzpicture}[scale=.6]
\begin{scope}[xshift=6cm]
\node[draw,black, fill=blue, circle, scale=1](b7) at (0, 7) {};
\node[draw,black, fill=yellow, circle, scale=1](b6) at (0, 6) {};
\node[draw,black, pattern=north west lines, pattern color=black, circle, scale=1](b5) at (0, 5) {};
\node at (0, 0) {$|W| = (1 -R)n$};

\draw[red] (-.5,6.4) arc (210:130:1);
\node[red] at (0,8) {Degree $s$};

\draw [decorate,decoration={brace,amplitude=10pt},xshift=-4pt,yshift=0pt]
(.5,7.5) -- (.5, 4.5) node [black,midway,xshift=0.4cm, anchor=west] 
{\begin{minipage}{7cm}\begin{center} Each ``layer'' is a random $(1,s)$-regular graph with $n$ left vertices and $n/s$ right vertices.  There are $t = (1-R)s$ layers.  \end{center}\end{minipage}};

\node[black, anchor=west](a) at (2,2) {\begin{minipage}{7cm} \begin{center} Each vertex $j$ imposes a \em parity check \em on its neighbors: $\sum_{i \in \Gamma(j)} \alpha_{i,j} c_i = 0$.  We choose the coefficients $\alpha_{i,j} \in \F_q^*$ uniformly at random. \end{center}\end{minipage}};
\draw[->,black](a) to (b5);

\foreach \i in {2,3,4}
{
\node[draw,black, circle, scale=1](b\i) at (0, \i) {};
\draw(b\i) to (-.7, \i + .3);
\draw(b\i) to (-1, \i );
\draw(b\i) to (-.7, \i - .3);
}
\end{scope}
\begin{scope}
\node at (0, 0) {$|V| = n$};
\foreach \i in {9,4,6}
{
\node[draw,black, fill=blue, circle, scale=1](a\i) at (0, \i) {};
\draw (b7) -- (a\i);
}
\foreach \i in {2,3,8}
{
\node[draw,black, fill=yellow, circle, scale=1](a\i) at (0, \i) {};
\draw (b6) -- (a\i);
}
\foreach \i in {5,7,1}
{
\node[draw,black, pattern=north west lines, pattern color=black, circle, scale=1](a\i) at (0, \i) {};
\draw (b5) -- (a\i);
}
\end{scope}
\end{tikzpicture}
\end{center}
\caption{A random $(t,s)$-regular bipartite graph that gives rise to a random $s$-LDPC code of rate $R$.  Here, we set $t:=s(1-R)$.}\label{fig:graph_pic}
\end{figure}

Our main theorem about the list-decodability of random LDPC codes is a reduction from the list-decodability of random linear codes:
\begin{theorem}\label{thm:list-dec}
For any $R \in (0,1)$, $\epsilon>0$, prime power $q$, $\alpha \in (0,1-1/q)$ and $L \geq1$ there exists $s_0 = s_0(\epsilon,\alpha,q,L) \geq 1$ such that the following holds
for any odd $s \geq s_0$.
Suppose that a random linear code of rate $R$ over $\F_q$ is $(\alpha,L)$-list decodable with high probability.
Then a random $s$-LDPC code of rate $R-\epsilon$ over $\F_q$ is $(\alpha,L)$-list decodable with high probability.
\end{theorem}
\begin{remark}[The parity of $s$]\label{rem:parity}
All of our results hold for even $s$ as well as odd $s$.  However, the proof is slightly simpler for odd $s$, so for clarity we state and prove the theorem in this case.
\end{remark}

\begin{remark}[Dependence of $s_0$]
It can be seen (see Remark~\ref{rmk:s1}) that we may take
\[ s_0 = O\left( \frac{ L \log q + \log(q/\eps) } { h_q^{-1}(1 - h_q(\alpha) - 1/L) } \right). \]
While this is not the focus of our work, it would be interesting to understand how large $s_0$ must be for a statement like Theorem~\ref{thm:list-dec} to hold. It is reasonable to suspect at least that $s_0$ must grow with $\eps$. As evidence for this suspicion, it is known (\cite{Gal62}) that for binary LDPC codes to be $\eps$-close to achieving the Gilbert-Varshamov bound,\footnote{The \deffont{GV bound} refers to the rate-distance trade-off $R = 1 - h_q(\delta)$, which is approached by a random linear code.} $s_0$ must grow with $\eps$. 
\end{remark}

Instantiating this with a result of \cite{GHK11} on list decoding of random linear codes, we get the following corollary.
\begin{corollary}\label{cor:list-dec-instantiated}
For any prime power $q$, $\alpha \in (0,1 - 1/q)$, and $\eps \in (0, 1 - h_q(\alpha))$ 
there exists  $L = O_\alpha(1/\eps)$ and $s\geq 1$ so that 
a random $s$-LDPC code of rate $1 - h_q(\alpha) - \eps$ over $\F_q$ is $(\alpha, L)$-list-decodable with high probability.
\end{corollary}

\begin{remark}[Other parameter regimes]
We state Corollary~\ref{cor:list-dec-instantiated} as one example of what can be obtained by combining Theorem~\ref{thm:list-dec} with one result on random linear codes.
The result of \cite{GHK11} degrades as $\alpha \to 1 - 1/q$, and so Corollary~\ref{cor:list-dec-instantiated} degrades as well.  However, there has been a great deal of work on the list-decodability of random linear codes as $\alpha \to 1-1/q$ (summarized in Section~\ref{sec:related} below), and Theorem~\ref{thm:list-dec} implies that these results carry over to random LDPC codes as well.  
\end{remark}

\paragraph{(2) Random LDPC codes achieve any local property that random linear codes achieve.}
Theorem~\ref{thm:list-dec} follows as a corollary of a much more general theorem.
We show that any ``local'' property that is satisfied by random linear codes with high probability is also satisfied by random LDPC codes with high probability.  

By a property $P_n$ of length $n$ codes over $\Sigma$, we mean a family $P_n \subseteq 2^{\Sigma^n}$ of codes in $\Sigma^n$, and we say that a code $C \subseteq \Sigma^n$ satisfies the property $P_n$ if $C \in P_n$. Informally, a local property is a property which can be defined by the exclusion of certain bad sets.  For example, a code $\cC$ is $(\alpha, L)$-list-decodable if it does \em not \em contain any sets $B \subset \Sigma^n$ of size larger than $L$ so that $B$ is contained in a Hamming ball of radius $\alpha$.  
Along with list-decodability, local properties include many related notions
like \em list recovery, \em  \em average-radius list decoding, \em  and \em erasure list decoding. \em 
A long line of work (discussed more in Section~\ref{sec:related}) has established that these properties hold for random linear codes with high probability, so our reduction immediately implies that they hold with high probability for LDPC codes as well.

Formally, we define a local property as follows.
Let $\pi:[n] \to [n]$ be a permutation on $[n]$. For a string $x \in \Sigma^n$, we let $\pi(x) \in \Sigma^n$ denote the string obtained by permuting the coordinates of $x$ according to $\pi$, and for a subset $B \subseteq \Sigma^n$, we let $\pi(B):=\{\pi(x)\mid x\in B\}$. We say that a collection $\B$ of subsets of $\Sigma^n$ is \deffont{permutation invariant} if for any $B \in \B$ and permutation $\pi:[n] \to [n]$, we also have that 
 $\pi(B) \in \B$.

\begin{definition}[Local property]\label{def:local}
Let $\cP=\{P_n\}_{n\in \N}$, where each $P_n$ is a property of length $n$ codes over $\Sigma$. We say that $\cP$ is a \deffont{$b$-local property} if for any $n \in \N$ there exists a permutation-invariant collection $\B_n$ of subsets of $\Sigma^n$, where $|B|\leq b$ for all $B \in \B_n$, such that
\begin{center}
$C \subseteq \Sigma^n$ satisfies $P_n$ $\iff$  $B \nsubseteq C$ for all $B \in \B_n$.
\end{center}
\end{definition}

We say that a family of random codes $C = \{C_{n_i}\}_{i \in \N}$ (where $\{n_i\}$ is an increasing sequence) \deffont{satisfies} $\cP$ \deffont{with high probability} if
$\lim_{i \to \infty} \Pr[C_{n_i} \; \text{satisfies}\; P_{n_i}] =1$. Similarly, we say that $C$ \deffont{almost surely does not satisfy $\cP$} if $\lim_{i \to \infty} \Pr[C_{n_i} \; \text{satisfies}\; P_{n_i}] =0$. 

A code property is \deffont{monotone decreasing} if given a code $C$ satisfying $P$, it holds that every code $C'\subseteq C$ also satisfies $P$. Note that every local property is monotone decreasing. 

A \deffont{random linear code} of rate $R$ over $\F_q$ is defined\footnote{There are a few natural ways to define a random linear code: for example we could also define it as a uniformly random subspace of dimension $Rn$, or we could define it as the image of a uniformly random $n \times Rn$ matrix, or we could define it as we do here, as the kernel of a uniformly random $(1-R)n \times n$ matrix.  It can be shown that these distributions are quite close to each other, and in particular, any property that holds for one with high probability holds for the others.} as the kernel of a uniformly random matrix $H \in \F_q^{(1-R)n \times n}$.  Notice that such a code has rate $R$ with high probability. 

For any $n\in \N$ and $R\in [0,1]$ such that $R\cdot n\in \N$, we denote a random linear length $n$ code of rate $R$ by $\CRLC^n(R)$. Likewise, given $s$, $n$ and $R$ such that $s\divides n$ and $R \cdot s\in \N$, we denote a random $s$-LDPC code of length $n$ and rate $R$ by $\CLDPC{s}^n(R)$. Whenever we use these notations, it is implicitly assumed that the relevant divisibility conditions are satisfied. 

Let $\cP = \{P_n\}_{n \in \N}$ be a monotone decreasing property of linear codes.  We define
\begin{equation}\label{eq:defthresh}
	\TRLC^n(\cP) := \begin{cases} \sup\inset{R \in [0,1]:\Pr[\CRLC^n(R) \; \text{satisfies} \; P_n] \geq 1/2} & \text{if there is such an $R$ } \\ 0 & \text{otherwise.} \end{cases}
\end{equation}

\begin{remark}
	If $\cP$ is a monotone decreasing property then the function $\Pr[\CRLC^n(R) \; \text{satisfies} \; P_n]$ is monotone decreasing in $R$. This can be proved by a standard coupling argument, akin to \cite[Thm.\ 2.1]{Bollobas11}.
\end{remark}

With the notation out of the way, we are ready to state our more general theorem about random LDPC codes.
Essentially, this theorem says that every local property that holds with high probability for a random linear code also holds with high probability for a random $s$-LDPC code of approximately the same rate. This approximation improves as $s$ grows. 

\begin{restatable}[Main]{theorem}{mainTheorem}\label{thm:main}
	Let $\cP = (P_n)_{n\in \N}$ be a $b$-local property with $\bar R := \limsup_{n\to \infty}\TRLC^n(\cP) < 1$. For any $\eps>0$ and prime power $q$, there exists $s_0=s_0(\eps,\bar R,q,b) \geq 1$ such that for any odd $s \geq s_0$ and any sequence $\{R_{n}\}_{n\in \N}$, if $R_{n}\le \TRLC^{n}(\cP)-\epsilon$ for all $n$, then the code ensemble $\CLDPC{s}^n(R_n)$ satisfies $\cP$ with high probability.
\end{restatable}

\begin{remark}[The dependence on $\eps, \bar R, q, b$]\label{rmk:s1} An inspection of the proof shows that we may take
	\[ s_0 = O\inparen{ \frac{b \log(q) + \log(q/\eps) }{ h_q^{-1}(1 - \bar R) }  }. \]
	In more detail, there are two parts of the proof that require $s \geq s_0(\eps,\bar R, q, b)$ to be sufficiently large: first, when we apply Lemma~\ref{lem:ProbabilityOfMatrixInLDPC}; and secondly, when we apply Theorem~\ref{thm:ldpc_distance}. Remark~\ref{rmk:s2} will state that the application of Lemma~\ref{lem:ProbabilityOfMatrixInLDPC} requires $s_0 \geq C_0 \cdot \frac{b \ln q}{\delta}$ for some constant $C_0$. For the application of Theorem~\ref{thm:ldpc_distance}, Remark~\ref{rmk:s3} will state that $s_0 \geq C_1\cdot\frac{\ln(q/\eps)}{\delta}$, for some constant $C_1$, suffices.
\end{remark}

The existence of a reduction like the one in Theorem~\ref{thm:main} is surprising, at least to the authors. There is a lot more structure in a random LDPC code than in a random linear code.  For example, we know of linear-time unique decoding algorithms for random LDPC codes,\footnote{This follows, for example, from \cite{SipserS94} because the underlying random graph is with high probability a good expander.} but it is unlikely that any efficient unique decoding algorithm exists for random linear codes.\footnote{Unique decoding of random linear codes is related to the problems of Learning Parities with Noise (LPN) and Learning With Errors (LWE), which are thought to be hard.}  Thus it is unexpected that this much more structured ensemble would share many properties---in a black-box way---with random linear codes.
 
\begin{remark}[A converse to Theorem \ref{thm:main}?]\label{rem:ConverseMain}
One may be tempted to conjecture that the converse of Theorem \ref{thm:main} holds as well.  Namely, in the setting of Theorem \ref{thm:main}, if $R_{n_i}\ge \TRLC^{n_i}(\cP)+\epsilon$ for all $i$, then the code ensemble $\CLDPC{s}(R_n)$ almost surely does not satisfy $\cP$. However, this turns out to be false, due to the following example.  Assume that $q=2$ and consider the $1$-local property $\cP := (P_n)_{n\in \N}$, where $P_n$ is the set of all length $n$ linear codes that only contain even weight codewords. It is not hard to see (e.g., using Theorem \ref{thm:GeneralRLCThreshold}) that $\TRLC^{n}(\cP)$ tends to $0$ as $n\to \infty$. On the other hand, if $\frac ns$ is even, then every $s$-LDPC code (including, say, a code of rate $\frac 12$) satisfies $\cP$, contradicting this conjecture.

However, the above counter-example relies on a technicality involving divisibility criteria. 
It is an interesting question whether a natural 
converse of Theorem \ref{thm:main} holds if we additionally assume that $\cP$ belongs to some natural class of ``nicely behaved'' properties that precludes counter-examples of this sort.
\end{remark}

\begin{remark} [Non-local properties]
	While local properties do indeed capture many natural coding-theoretic properties, it does not capture them all. For example, it is unclear to us how to capture dual distance, i.e., the minimum weight of a non-zero parity-check satisfied by a linear code; or the covering radius, i.e., the minimum radius $r\geq 0$ such that Hamming balls of radius $r$ centered at codewords cover all of $\Sigma^n$.
\end{remark}

\paragraph{(3) A characterization of local properties satisfied by random linear codes.}
In order to prove Theorems~\ref{thm:list-dec} and \ref{thm:main}, we develop a new characterization of the local properties satisfied by a random linear code.  Our formal theorem is given as Theorem~\ref{thm:GeneralRLCThreshold}.  Informally, this theorem implies that for any monotone decreasing property $\mathcal{P}$, there is a sharp threshold $R^*$ so that random linear codes of rate slightly less than $R^*$ with high probability satisfy $\mathcal{P}$, while random linear codes of rate slightly larger than $R^*$ with high probability do not.  Moreover, we give a characterization of $R^*$.  

Formally, we have the following definition, recalling the definition of $\TRLC^n(R_n)$ from \eqref{eq:defthresh}.
\begin{definition}[Sharpness for random linear codes]
We say that the property $\cP$ is \deffont{sharp for random linear codes} if for every $\epsilon>0$ there holds:
\begin{itemize}
	\item If $R_n \leq \TRLC^n(\cP) - \eps$ for large enough $n$, then the code ensemble $\CRLC^n(R_n)$ ($n\in \N$) satisfies $\cP$ with high probability. 
	\item If $R_n \geq \TRLC^n(\cP) + \eps$ for large enough $n$, then the code ensemble $\CRLC^n(R_n)$ ($n\in \N$) almost surely does not satisfy $\cP$. 
\end{itemize}
\end{definition}
%
\noindent If a property $\cP$ is sharp, we sometimes refer to $\TRLC^n(\cP)$ as the \emph{threshold} for $\cP$. 
\vspace{.3cm}

Theorem~\ref{thm:GeneralRLCThreshold} has two corollaries.
The first is that local properties are sharp for random linear codes: 

\begin{corollary}\label{cor:SharpnessForRLC}
Every local property is sharp for random linear codes.
\end{corollary}

The second corollary of Theorem~\ref{thm:GeneralRLCThreshold} is a characterization of $\TRLC^n(\cP)$.  This characterization requires some definitions to state formally, so we defer the formal statement to Theorem~\ref{thm:GeneralRLCThreshold}.  However, it has an intuitive interpretation, which we sketch here.

Recall that a local property is defined by a permutation-invariant collection $\mathcal{B}_n$ of excluded sets.  For simplicity of exposition, suppose that all of the sets $B \in \mathcal{B}_n$ have size exactly $b$, and moreover that they all have dimension exactly $b$.  (This assumption is helpful for exposition but not necessary for our analysis).  In this case, it is easy to compute the probability that each individual set $B \in \mathcal{B}_n$ is contained in $\CRLC(R)$
(see Fact~\ref{fact:random_linear_matrix}):
\[ \PR{ B \subseteq \CRLC(R) } = q^{-(1-R)nb}.\]
Thus, we have
\[ \mathbb{E} \inabs{ \inset{ B \in \mathcal{B}_n \suchthat B \subseteq \CRLC(R) } } = |\mathcal{B}_n| \cdot q^{-(1-R)nb}. \]
Thus, as long as
\[ R < \TRLC^{\mathbb{E}}(\mathcal{B}_n) := 1 - \frac{\log|\mathcal{B}_n|}{nb}, \]
we are guaranteed by Markov's inequality that with high probability, no elements of $\mathcal{B}_n$ appear in $\CRLC(R)$.  However, what if $R > \TRLC^{\mathbb{E}}(\mathcal{B}_n)$?  It turns out that the statement above is not tight: in some cases it is likely that no elements of $\mathcal{B}_n$ appear in $\CRLC(R)$ even if the rate $R$ is significantly larger than $\TRLC^{\mathbb{E}}(\mathcal{B}_n)$.  We give an example in Example~\ref{ex:ExpecrationBoundNotTight} of when this can occur.

Our result in Theorem~\ref{thm:GeneralRLCThreshold} pins down exactly when this can occur.  Informally, it happens only because some projection $\mathcal{B}_n'$ of the collection $\mathcal{B}_n$ is more favorable than one might expect, in the sense that $\TRLC^{\mathbb{E}}(\mathcal{B}_n')$ is larger than one might expect.  In this case, the ``correct'' threshold is precisely $\TRLC^{\mathbb{E}}(\mathcal{B}_n')$.  

Thus, Theorem~\ref{thm:GeneralRLCThreshold} also provides a characterization of which sorts of ``bad'' lists $B$ (up to a permutation of the coordinates) are contained in a random linear code of a particular rate.  We hope that this characterization will be useful in the study of random linear codes themselves, in addition to random LDPC codes.  

The full power of Theorem~\ref{thm:GeneralRLCThreshold} (including the characterization of $\TRLC^n(\cP)$ described above) is used to prove Theorem~\ref{thm:main}. 
However, given Theorem~\ref{thm:main},
Theorem~\ref{thm:list-dec} readily follows from Corollary~\ref{cor:SharpnessForRLC} itself: 
\begin{proof}[Proof of Theorem \ref{thm:list-dec}]
Let $\cP$ denote the property of being $(\alpha, L)$-list-decodable. Note that $\cP$ is a local property: for any $n \in \N$, take $\B_n$ to be the collection of all sets of $L+1$ vectors in $\F_q^n$ contained in some Hamming ball of radius $\alpha$. Now, fix some $R\in (0,1)$ and assume that a random linear code of rate $R$ satisfies $\cP$ with high probability. Corollary \ref{cor:SharpnessForRLC} implies that $\TRLC^n(\cP) \le R+o_{n\to \infty}(1)$. 


Next, it is not hard to verify that $\limsup_{n\to \infty}\TRLC^n(\cP)\le 1-h_q(\alpha)<1$. Indeed, it follows from the list-decoding capacity theorem (e.g.\ \cite[Thm 1.1]{LiW18}) that for large enough $n$ there are no $(\alpha,L)$-list-decodable codes of rate $1-h_q(\alpha)+\epsilon$. In particular, this means that a random linear code of rate $1-h_q(\alpha)+\epsilon$ almost surely does not satisfy $\cP$.

Theorem~\ref{thm:main} now immediately yields Theorem~\ref{thm:list-dec}.
\end{proof}

We give a high-level overview of the proof of Theorem~\ref{thm:main} in Section~\ref{sec:mainthm} below after a discussion of related work in Section~\ref{sec:related}. 

\subsection{Related Work}\label{sec:related}

\paragraph{List-decodability of random ensembles of codes.}
As mentioned above, 
it is not hard to see that a completely random code $\cC \subset \Sigma^n$ achieves list-decoding capacity.
There has also been work studying more structured random ensembles of codes, notably random linear codes.  
Zyablov and Pinsker~\cite{ZyablovP81} showed that random linear codes of rate $1 - h_q(\alpha) - \eps$ are $(\alpha, L)$-list-decodable with high probability, where $L$ is independent of $n$ but depends exponentially on $1/\eps$.
Two decades later, \cite{GuruswamiHSZ02} showed that there exist binary linear codes with list-size $O(1/\eps)$, and their techniques were recently extended to hold with high probability in \cite{LiW18}.
In the meantime, \cite{GHK11} showed that random linear codes over any constant-sized alphabet achieve capacity with $L = O(1/\eps)$ when $\alpha$ is bounded away from $1 - 1/q$; \cite{CheraghchiGV13, Wootters13, RudraW14a,RudraW18} extended these results to get list sizes nearly as good even for large $\alpha$, although the problem is still open in some parameter regimes.

Several variants of list-decoding have been studied for random linear codes, including \em list-recovery\em~\cite{RudraW18}, \em average-radius list-decoding\em~\cite{Wootters13, RudraW14a, RudraW18}, and list-recovery from erasures~\cite{Guruswami03}.\footnote{
List-recovery is a generalization of list-decoding where the input is a list of sets $Z_1, \ldots, Z_n$ of size at most $\ell$ (instead of a received word $z \in \Sigma^n$, which can be seen as the $\ell=1$ case), and goal is to find all of the codewords $c \in C$ so that $c_i \in Z_i$ for at least a $1 - \alpha$ fraction of the $i \in [n]$.  Average-radius list-decoding is a strengthening of list-decoding where instead of requiring that no set of $L+1$ codewords are \em all \em close to some $z$, we require that no set of $L+1$ codewords has small \emph{average} distance to $z$.  List-decoding from erasures is a weaker notion than list-decoding, where $z \in (\Sigma \cup \{\bot\})^n$ has some \em erased \em symbols, and the goal is to recover all $c \in \cC$ which agree with $z$ on the observed coordinates.}
All of these properties are local, and so our main theorem implies that LDPC codes satisfy them with high probability.  

\paragraph{List-decodability of explicit codes.}
Obtaining explicit constructions of codes which achieve list-decoding capacity was a major open problem until it was solved about a decade ago.  The first explicit codes to provably achieve capacity were the \em Folded Reed-Solomon Codes \em of Guruswami and Rudra~\cite{GuruswamiR08}.  These codes are variants on the classic \em Reed-Solomon codes \em and are based on polynomials over finite fields.
Since then, there have been several constructions of such codes, also based on algebraic techniques, including \em Univariate Multiplicity Codes \em \cite{GuruswamiW13a, Kopparty15, KoppartyRSW18}, variants of Algebraic-Geometry Codes~\cite{GuruswamiX12,GuruswamiX13}, and manipulations of these codes~\cite{DvirL12, GuruswamiK16, HemenwayRW17, KRRSS19}. 
However, 
the state-of-the-art for explicit constructions still requires quite large (but constant) alphabet and list sizes.  These codes can be efficiently list-decoded in polynomial time; the fastest algorithm is that of \cite{HemenwayRW17,KRRSS19}, which runs in nearly-linear time $O(n^{1 + o(1)})$.

While graph-based techniques have been used to modify the underlying algebraic constructions (for example the expander-based distance-amplification technique of \cite{AEL} is used in \cite{HemenwayRW17,KRRSS19} to obtain near-linear-time list-decoding), to the best of our knowledge there are no results establishing list-decodability up to capacity for purely graph-based codes such as LDPC codes or expander codes.\footnote{We note that \cite{HemenwayW15} give capacity-achieving graph-based codes for zero-error list-recovery (with erasures), where the input is lists $Z_1,\ldots, Z_n$ so that most lists have small size, and the goal is to return all codewords $c \in \cC$ that satisfy $c_i \in Z_i$ for all $i$. It does not seem easy to adapt these techniques for general list-recovery and hence for list-decoding. } 

Finally, we note that recent work~\cite{DHKTLS19} has given an algorithm to list-decode codes based on high-dimensional expanders, but these results are far from list-decoding capacity.

\paragraph{LDPC Codes Achieve Capacity on the Binary Symmetric Channel.}
LDPC Codes have been studied extensively in the context of unique decoding, especially in a model of random errors.  
Informally, a code is said to \deffont{achieve capacity on the Binary Symmetric Channel} (BSC) if there is some algorithm which can, with high probability, uniquely decode a code of rate $R = 1 - h_2(\alpha) - \eps$ from an $\alpha$-fraction of \em random \em errors.  It is known that Gallager's LDPC codes nearly achieve capacity on the BSC as $n$ gets large, under maximum-likelihood decoding~\cite{Gal62,venkat-survey}, and recently it was shown that certain LDPC codes achieve capacity for smaller block lengths under efficient decoding algorithms as well~\cite{SCLDPC}.  
Achieving capacity on the BSC is related to achieving list-decoding capacity (in particular, the capacities are the same, $R = 1 - h_q(\alpha)$).  However, there is no formal connection along these lines, and to the best of our knowledge these results about the BSC do not imply anything about the list-decodability of LDPC codes.

\paragraph{Relationship to threshold results in combinatorics.} 
Finally, we note that our results providing sharp thresholds of local properties for random linear codes are 
reminiscent of classic results about local properties of random graphs.  We discuss this connection more in Remark~\ref{rem:relationship_to_graphs}.
We note that, due to the difference in setting and parameter regime, our use of the word ``sharp'' does not exactly line up with the definition of a sharp threshold in graph theory.  In particular, as we focus on constant rate codes, we do not prove results about the width of the threshold for $k = o(n)$.

For thresholds for random subspaces, the recent independent work of Rossman~\cite{rossman} shows a statement similar to our Corollary~\ref{cor:SharpnessForRLC}.  More precisely, that work establishes the existence of sharp thresholds for monotone properties of random subspaces.  That work uses completely different methods from ours.  In particular, the proof establishes the existence of such thresholds but does not imply the characterization that we find in our work for local properties.  This characterization is key for our application to LDPC codes.

\subsection{Discussion and open questions}\label{sec:discussion}
In this work, we answer Question~\ref{q:main} with a very strong ``yes.'' There are LDPC codes that achieve list-decoding capacity, and moreover there are many of them, and moreover these codes also likely satisfy any local property---that is, any property which can be defined by ruling out small bad sets of codewords---which is likely satisfied by a random linear code. Our results raise several interesting questions:

\begin{enumerate}
\item \textbf{What other properties are local?}  We have shown that random LDPC codes satisfy with high probability any local property that random linear codes satisfy with high probability.  There are several natural examples of local properties, including distance, list-decoding and list-recovery.  What other examples are there?
\item \textbf{What other applications of Theorem~\ref{thm:GeneralRLCThreshold} are there?}
In subsequent work \cite{GLMRSW20}, the characterization of a sharp threshold for local properties of random linear codes (Theorem~\ref{thm:GeneralRLCThreshold}) was already demonstrated to be useful beyond our work on LDPC codes. We hope to see additional applications of this result. For example, Remark~\ref{rmk:RLCThersholdProb} implies that
to prove that $\CRLC(R - \eps)$ satisfies a local property $\cP$ with probability $1 - 2^{-\Omega(n)}$, it suffices to show that $\CRLC(R)$ satisfies $\cP$ with some tiny probability (at least $2^{-o(n)}$).  Are there situations where this could be useful? 
\item \textbf{Derandomization?}  Our results hold for a random ensemble of LDPC codes. It is natural to ask whether (or to what extent) this construction can be derandomized. In particular, it does not seem as though the underlying graph being an expander would be sufficient.
\item \textbf{Algorithms?}  Our results are combinatorial, but one of our main motivations is algorithmic.  At the moment we do not know of any truly linear-time list-decoding algorithms for any capacity-achieving list-decodable codes.  Since essentially all known linear-time algorithms in coding theory arise from graph-based codes,
such codes are a natural candidate for linear-time list-decoding.  Now that we know that random LDPC codes achieve list-decoding capacity combinatorially, can we list-decode them efficiently?
\end{enumerate}

\subsection{Organization and main building blocks}\label{sec:organization}
In Section~\ref{sec:mainthm}, we give a high-level overview of the proof of   Theorem~\ref{thm:main}. This proof relies on three building blocks:

\begin{itemize}
\item  First, Lemma~\ref{lem:tau-threshold} 
establishes sharp thresholds for certain local properties, and effectively characterizes
the sorts of sets $B\subseteq \F_q^n$ that are contained in a random linear code. We prove this lemma in Section~\ref{sec:random_linear_matrix_char}. Using Lemma~\ref{lem:tau-threshold} we prove Theorem~\ref{thm:GeneralRLCThreshold}, which pins down a sharp threshold for any local property of a random linear code. 
\item Second, Lemma~\ref{lem:ProbabilityOfMatrixInLDPC} shows that for a set $B$ with a certain property called \emph{$\delta$-smoothness}, the probability that $B$ appears in a random $s$-LDPC code is not much larger than the probability that it appears in a random linear code of the same rate. 
We prove this Lemma~\ref{lem:ProbabilityOfMatrixInLDPC} in Section~\ref{sec:random_ldpc_matrix} using Fourier analysis.

Together with Lemma~\ref{lem:tau-threshold}, Lemma~\ref{lem:ProbabilityOfMatrixInLDPC} implies that any property satisfied with high probability by a random linear code is also satisfied with high probability by a random $s$-LDPC code of similar rate, provided that we can restrict our attention to $\delta$-smooth sets $B$.  
It turns out that for any code with good distance,\footnote{The \deffont{distance} of a code is the minimum distance between any two codewords.} we may indeed restrict our attention to such sets, so it remains to show that random $s$-LDPC codes have good distance.

\item Third, Theorem~\ref{thm:ldpc_distance} shows that random $s$-LDPC codes do indeed have good distance with high probability.  This was already shown by Gallager in the binary case; we give an alternative proof of this fact that also extends to large alphabets.  We prove Theorem~\ref{thm:ldpc_distance} in Section~\ref{sec:distance} using techniques from exponential families.
\end{itemize}

Together, these three building blocks can be used to establish Theorem~\ref{thm:main}, as we show next in Section~\ref{sec:mainthm}.

%% file: high_level.tex
In this section we prove our main theorem (Theorem~\ref{thm:main}) using the building blocks outlined in Section~\ref{sec:organization}.  We will establish these building blocks in later sections.  
The purpose of this section is to give a high-level idea of the structure of the proof, deferring the technical parts to later sections.  However, we will need a few technical definitions, outlined in Section~\ref{sec:notation}.

\subsection{Notation and definitions} \label{sec:notation}
Because we are studying local properties, we need some notation around sets $B \subseteq \F_q^n$.
For such a set $B$ of size $\ell$, it will be convienient to view $B$ as a matrix $M \in \F_q^{n \times \ell}$ with the elements of $B$ as the columns. (The ordering of the columns will not matter.)  We say that $M$ is \deffont{contained} in a code $\cC \subseteq \F_q^n$ (written ``$M \subset \cC$'') if all of the columns of $M$ belong to $\cC$.

The notion of permutation-invariant properties leads us to think about permutations of the \em rows \em of such a matrix $M \in \F_q^{ n \times \ell}$.
Motivated by this, we define $\tau_M$, the \deffont{row distribution} of $M$, as follows:  
for any $v \in \F_q^\ell$,
\[ \tau_M(v):= \frac{\text{number of appearances of $v$ as a row in $M$}} {n}. \] 

Let $\cD_{n,\ell}$ denote the collection of possible row distributions of matrices in $\F_q^{n \times \ell}$, \em i.e., \em distributions $\tau$ over $\F_q^\ell$ where $\tau(v) \cdot n \in \N$ for any $v \in \supp(\tau)$.\footnote{Notice that $\cD_{n,\ell}$ depends on $q$ as well, but we suppress this dependence in the notation for readability.}
The number of possible row distributions of matrices in $\F_q^{n \times \ell}$ is just the number of ways to partition $n$ things into at most $q^\ell$ groups, so 
\begin{equation}\label{eq:sizeD}
|\cD_{n,\ell}| \leq { n + q^\ell - 1 \choose q^\ell - 1 }.
\end{equation}
For
a distribution $\tau \in \cD_{n,\ell}$, let $\cM_{n,\tau}$ denote the collection of matrices $M \in \F_q^{n \times \ell}$ with row distribution $\tau$. We say that a code $C$ \emph{contains $\tau$} to mean that $M\subset C$ for some matrix $M\in\cM_{n,\tau}$. Let $$\cL_\tau = \{n\in \N \mid \tau(u)\cdot n\text{ is an integer for all }u\in \F_q^\ell \}.$$ Note that for $C$ to contain $\tau$, a trivial necessary condition is that the length of $C$ belongs to $\cL_\tau$. Let $\cP_\tau$ denote the $\ell$-local property of not containing any matrix from the set $\cM_{n,\tau}$. Properties of the form $\cP_\tau$ are particularly useful to us due to the following observation:

\begin{observation}[Local property decomposition]\label{obs:LocalPropertyDecomposition}
Let $\cP=(P_n)_{n\in \N}$ be an $\ell$-local property for some $\ell\in \N$. Then, for every $n\in \N$ there exists $T_n \subseteq \cD_{n,\ell}$ such that
\begin{center}
$C\subseteq \F_q^n$  satisfies $P_n \iff C$ satisfies $P_{\tau}$ for all $\tau\in T_n$.
\end{center}
\end{observation}
\begin{proof}
Note that for every $\tau\in D_{n,\ell}$, the set of matrices $\cM_{n,\tau}$ is closed under row permutations. The lemma now follows immediately from the definition of a local property.
\end{proof}

Finally, let $H(\tau)$ and $H_q(\tau)$ denote the \deffont{entropy} and \deffont{base-$q$-entropy} of a random variable distributed according to $\tau$:
\[ H(\tau) := -\sum_{x \in \supp(\tau)} \tau(x) \log( \tau(x))  \qquad \text{and} \qquad  H_q(\tau) := \frac{ H(\tau) }{\log q}.\]
Let
\[ d(\tau) := \dim(\spn(\supp(\tau))). \]

We will work with the \deffont{parity-check matrix} view of a random $s$-LDPC code $\cC$.  
Let $H \in \F_q^{(1-R)n \times n}$ be the adjacency matrix of the graph $G$ in Figure~\ref{fig:graph_pic} where the nonzero entries are given by the coefficients $\alpha_{i,j}$ of the parity checks.  Then we can define a random $s$-LDPC code $C$ as 
\[ C = \inset{ x \in \F_q^n \suchthat H \cdot x = 0 }. \]
We introduce some notation to talk about the structure of $H$, which we will use throughout the paper.  This is illustrated in Figure~\ref{fig:F}.

Let $F \in \{0,1\}^{(n/s) \times n}$ be the matrix 
$F=(F_1 \mid F_2 \mid \ldots \mid F_{n/s}),$
where each $F_i \in \{0,1\}^{(n/s)\times s}$ has all-ones $i$-th row, and the rest of the rows are all-zeros.
 Let $\Pi \in \{0,1\}^{n \times n}$ be a random permutation matrix, and let $D \in \F_q^{n \times n}$ be a diagonal matrix with diagonal entries that are uniform in $\F_q^*$. 
Let $H_1, \ldots, H_{(1-R)\cdot s}$ be sampled independently according to the distribution $F \cdot \Pi \cdot D$.  Then let $H \in \F_q^{(1-R)n \times n}$ be the matrix obtained by stacking $H_1, \ldots, H_{(1-R)\cdot s}$ on top of each other (see Figure~\ref{fig:F}).    Then $H$ is the parity-check matrix for a random $s$-LDPC code of rate $R$.
We will refer to each $H_i$ as a ``layer'' of $H$.
\begin{figure}
\centering
\begin{tikzpicture}[yscale=0.4,xscale=0.6]
\begin{scope}
\node at (-2,-2) {$F=$};
\draw (0,0) rectangle (10,-5);
\foreach \i in {0,2,4,6,8}
{
\node[anchor=west] at (\i, -\i/2-0.5) {$111111$};
}
\draw [decorate,decoration={brace,amplitude=5pt, mirror},xshift=-4pt,yshift=0pt]
(0,0) to (0,-5);
\node at (-.8, -2.5) {\large $\frac{n}{s}$};
\draw [decorate,decoration={brace,amplitude=5pt},yshift=4pt,xshift=0pt]
(0,0) to (2,0);
\node at (1,1.4) {$s$};
\end{scope}
\begin{scope}[xshift=11.5cm]
\node at (0,-3) {$H=$};
\draw (1,0) rectangle (11,-7);
\node at (6, -.75) {$H_1$};
\draw (1,-1.5) to (11,-1.5);
\node at (6,-2.25) {$H_2$};
\draw (1,-3) to (11,-3);
\draw (1,-5.5) to (11, -5.5);
\node at (6, -6.25) {$H_{(1-R)\cdot s}$};
\node at (6, -4.25) {$\vdots$};
\draw [decorate,decoration={brace,amplitude=2.5pt,mirror},xshift=-4pt,yshift=0pt]
(1,0) to (1,-1.5);
\node at (.3, -.75) {$\frac{n}{s}$};
\draw [decorate,decoration={brace,amplitude=8pt},xshift=4pt,yshift=0pt]
(11,0) to (11,-7);
\node at (13, -3.5) {$(1-R)n$};
\draw [decorate,decoration={brace,amplitude=5pt},yshift=4pt,xshift=0pt]
(1,0) to (11,0);
\node at (6,1.4) {$n$};
\end{scope}
\end{tikzpicture}
\caption{The matrices $F$ and $H$.  Each layer $H_i$ of $H$ is drawn independently according to the distribution $F \cdot \Pi \cdot D$, where $\Pi \in \{0,1\}^{n \times n}$ is a random permutation and $D \in \F_q^{n \times n}$ is a diagonal matrix with diagonal entries that are uniform in $\F_q^*$. }  \label{fig:F}
\end{figure}
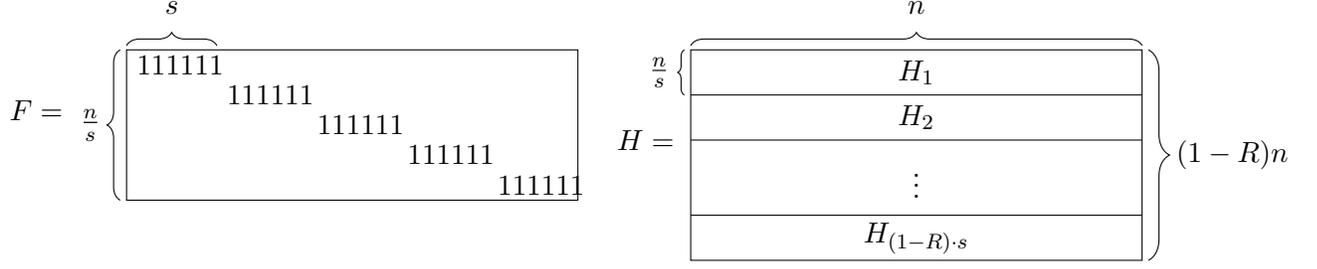

We will also require the following standard facts:

\begin{fact}\label{fact:random_linear_matrix}
	A matrix $M \in \F_q^{n \times \ell}$ is contained in a random linear code $C \subseteq \F_q^n$ of rate $R$ with probability $q^{-  (1-R) \cdot \rank(M) \cdot n}$.	
\end{fact}

We include the proof of Fact~\ref{fact:random_linear_matrix} for completeness.
\begin{proof} 
	Let $v_1, \ldots, v_{\rank(M)}$ be columns of $M$ that form a basis for the column span of $M$.  Then for each $v_i$, $\PR{ v_i \in \cC } = q^{-(1-R)n}$.  Since the $v_i$ are linearly independent, the events that they are contained in a random linear code $\cC$ are stochastically independent, and so the probability that all $\rank(M)$ of these vectors are contained in $C$ is $q^{-(1-R) \cdot \rank(M) \cdot n}$.  
\end{proof}

\begin{fact}[\cite{tutorial}, Lemma 2.2] \label{fact:orbit-size}
	For any distribution $\tau \in \cD_{n,\ell}$,
	$$q^{H_q(\tau) \cdot n} \cdot {n+q^\ell -1 \choose q^\ell -1}^{-1} \leq |\cM_{n,\tau}| \leq  q^{H_q(\tau) \cdot n}.$$
\end{fact}

Hence, when $\ell$ and $q$ are constants (which is the setting we investigate), we have $|\cM_{n, \tau}| \geq q^{H_q(\tau) \cdot n}/\poly(n)$. 

\subsection{Sharp thresholds for local properties for random linear codes}\label{subsec:random_linear_charac}

The first building block is Lemma~\ref{lem:tau-threshold} below, which shows that for every distribution $\tau\in\cD_{n,\ell}$, the property $\cP_\tau$ is sharp for random linear codes. Moreover we give a simple characterization of $\TRLC(\cP_{\tau})$. 
As an easy corollary, we get Theorem~\ref{thm:GeneralRLCThreshold}, which generalizes Lemma~\ref{lem:tau-threshold} to any local property, not necessarily of the form $\cP_\tau$.

Before stating Lemma~\ref{lem:tau-threshold} we give some intuition. Fix some distribution $\tau$ over $\F_q^\ell$. Let $C$ be a random linear code of length $n\in \cL_\tau$ and rate $R$. We seek a threshold rate, above which $C$ is likely to contain $\tau$. It is natural to attempt a first-moment approach to this problem and ask what is the expected number of matrices from $\cM_{n,\tau}$ which are contained in $C$. Note that $|\cM_{n,\tau}| = q^{n\cdot H_q(\tau)}\cdot \poly(n)$. Indeed, if $u_1,\dots,u_{q^\ell}$ are an enumeration of $\F_q^\ell$, then $\cM_{n,\tau}$ is in one-to-one correspondence with partitions on $[n]$ into $q^\ell$ subsets of sizes $n\tau(u_1),\dots,n\tau(u_{q^\ell})$. That is, $|\cM_{n, \tau}| = \binom{n}{n\tau(u_1),\dots,n\tau(u_{q^\ell})} = q^{nH_q(\tau)}\cdot\poly(n)$, where the last estimate follows from Fact~\ref{fact:orbit-size}, and relies on our assumption that $n\in \cL_\tau$.

Given $M\in \cM_{n,\tau}$, the code $C$ contains $M$ with probability $q^{-n\cdot (1-R)\cdot d(\tau)}$ (see Fact~\ref{fact:random_linear_matrix}). Hence, in expectation, $C$ contains roughly $q^{n\cdot (H_q(\tau) - (1-R)\cdot d(\tau))}$ matrices from $\cM_{n,\tau}$. In particular, this expectation grows (resp. decays) exponentially in $n$, when $R$ is larger (resp. smaller) than $1-\frac{H_q(\tau)}{d(\tau)}$. This motivates the following definition.

\begin{definition}[Expectation threshold]
	Given a distribution $\tau$ over $\F_q^\ell$, define the \emph{expectation-threshold} $$\TRLCEXP(\tau) := 1-\frac{H_q(\tau)}{d(\tau)}.$$
\end{definition}

It follows immediately from a first-moment argument that if $R < \TRLCEXP(\tau)$ then $C$ satisfies $\cP_\tau$ with probability $1-e^{-\Omega(n)}$. In particular, as $n$ grows we get the lower bound 
\begin{equation}\label{eq:ExpectationThresholdBound}
\TRLC^n(\cP_\tau) \ge \TRLCEXP(\tau)-o(1).
\end{equation} 
However, as the following example shows, this bound is not tight.

\begin{example} \label{ex:ExpecrationBoundNotTight}
	Let $q=2$, $\ell=3$ and consider the distribution $\tau$ over $\F_2^3$, given by the following table:
\begin{center}
	\begin{tabular}{ |c|c| } 
		\hline
		$u$&$\tau(u)$\\
		\hline\hline
		$(1,0,0)$ & $1/4$\\ 
		\hline
		$(0,1,0)$ & $1/4$\\ 
		\hline
		$(1,0,1)$ & $1/4$\\ 
		\hline
		$(0,1,1)$ & $1/4$\\ 
		\hline
		Every other vector & $0$\\
		\hline
	\end{tabular}
\end{center}
	It is straightforward to compute $\TRLCEXP(\tau) = 1-\frac{H_2(\tau)}{d(\tau)} = 1-\frac 23 = \frac 13$. 
	
	We claim that $\TRLC^n(\cP_\tau)$ is bounded away from $\TRLCEXP(\tau)$. Let $A:=\bigl( \begin{smallmatrix}
	1&0&0\\ 0&1&0
	\end{smallmatrix} \bigr)\in \F_2^{2\times 3}$ represent the linear map which projects a vector onto its first two coordinates. Let $\tau'$ denote the distribution of $Au$, where $u$ is a random vector sampled from $\tau$. Thus, $\tau'$ is distributed as follows:
\begin{center}
	\begin{tabular}{ |c|c| } 
		\hline
		$u$&$\tau'(u)$\\
		\hline\hline
		$(1,0)$ & $1/2$\\ 
		\hline
		$(0,1)$ & $1/2$\\ 
		\hline
		Every other vector & $0$\\
		\hline
	\end{tabular}	
\end{center}
	Note that a code $C$ which contains a matrix $M$ from $\cM_{n,\tau}$ must contain the first two columns of $M$: that is, the matrix $MA^T$. Consequently, every code which satisfies $\cP_{\tau'}$ also satisfies $\cP_{\tau}$, and so $\TRLC^n(\cP_{\tau}) \ge \TRLC^n(\cP_{\tau'})$.   
	
	Finally, (\ref{eq:ExpectationThresholdBound}) yields $$\TRLC^n(\cP_{\tau'}) \ge \TRLCEXP(\tau') -o(1) = 1- \frac{H_2(\tau')}{d(\tau')} -o(1)= 1-\frac12 -o(1)= \frac 12-o(1)$$
	and we conclude that
	$$\TRLC^n(\cP_{\tau}) \ge \frac 12-o(1)> \frac 13 = \TRLCEXP(\tau)$$
	for large $n$.
\end{example}

In Example~\ref{ex:ExpecrationBoundNotTight}, the bound of $\TRLCEXP(\tau)$ was not tight, in that the rate can actually be much higher than we would expect from a first-moment argument.  The reason was that there was some linear map $A$ so that $\tau' = A\tau$ had a larger value of $\TRLCEXP(\tau')$.  We will show below that this is the only reason that $\TRLCEXP(\tau)$ might not be the right answer.  To make this precise, we introduce the following definition.

\begin{definition}[Implied distribution]\label{def:ImpliedDistribution}
Let $\tau$ be a distribution over $\F_q^\ell$ and let $A\in \F_q^{m\times \ell}$ be a rank $m$ matrix for some $m\le \ell$. The distribution of the random vector $Au$, where $u$ is randomly sampled from $\tau$, is said to be \deffont{$\tau$-implied}. We denote the set of $\tau$-implied distributions by $\cI_\tau$.
\end{definition}

Note that whenever $\tau'\in \cI_\tau$, a linear code satisfying $\cP_{\tau'}$ must also satisfy $\cP_{\tau}$. Indeed, in the setting of Definition \ref{def:ImpliedDistribution} assume that $C$ contains a matrix $M\in \cM_{n,\tau}$. By linearity, $C$ also contains the matrix $MA^T$, which belongs to $\cM_{n,\tau'}$. Hence, not satisfying $\cP_{\tau}$ implies not satisfying $\cP_{\tau'}$. Consequently, $\TRLC^n(\cP_\tau) \ge \TRLC^n(\cP_{\tau'})$. 

Inequality (\ref{eq:ExpectationThresholdBound}) now yields the stronger bound
\begin{equation}\label{eq:ImpliedExpectationTheresholdBound}
\TRLC^n(\cP_\tau) \ge \max_{\tau'\in \cI_\tau} \TRLCEXP(\tau') - o(1).
\end{equation}
Lemma~\ref{lem:tau-threshold} below essentially says that (\ref{eq:ImpliedExpectationTheresholdBound}) is tight, and that $\cP_{\tau}$ is sharp for random linear codes. We prove this Lemma in Section \ref{sec:random_linear_matrix_char}.

\begin{restatable}[Sharp threshold for $\cP_\tau$ for random linear codes]{lemma}{sharpnessRandomLinear}\label{lem:tau-threshold}
	Let $\ell\in \N$ and let $\tau$ be a distribution over $\F_q^\ell$. Denote $R^*_\tau = \max_{\tau'\in \cI_\tau} \TRLCEXP(\tau')$. 
	Fix any $\eps>0$, and let $C$ be a random linear code of rate $R$ and length $n\in \cL_\tau$. The following holds:
	\begin{enumerate}[label=\textbf{(\roman*)},font=\upshape]
		\item\label{enum:SharpnessForRLCLowRate} If $R \leq R^*_\tau - \eps$, then
		\[
			\PR{ \exists M \in \cM_{n, \tau}, M \subset  C } \leq q^{-\eps n}.
		\]
		\item\label{enum:SharpnessForRLCHighRate} If $R \geq R^*_\tau + \eps$, then
		\[
			\PR{ \exists M \in \cM_{n,\tau}, M \subset C} \geq 1-{n+q^{2\ell} -1 \choose q^{2\ell} -1}^3 \cdot q^{-\epsilon n}.
		\]
	\end{enumerate}
\end{restatable}


%

Having established a sharp threshold for properties defined by excluding a single type, we can conclude a sharp threshold phenomenon for all local properties.

\begin{theorem}[Sharp thresholds for local properties for random linear codes]\label{thm:GeneralRLCThreshold}
Fix $\ell \in \N$. Let $\cP= (P_n)_{n\in \N}$ be an $\ell$-local property and let $(T_n)_{n\in \N}$ be as in Observation \ref{obs:LocalPropertyDecomposition}. Then $\cP$ is sharp for random linear codes and $$\TRLC^n(\cP)= \min_{\tau\in T_n}\max_{\tau'\in \cI_\tau}\TRLCEXP(\tau') \pm o_{n\to \infty}(1).$$
\end{theorem}

\begin{proof} [Proof of Theorem~\ref{thm:GeneralRLCThreshold}]
Denote $$R^*_n = \min_{\tau\in T_n}\max_{\tau'\in \cI_\tau}\TRLCEXP(\tau')$$ and fix $\epsilon > 0$. To prove the theorem, it suffices to show the following:
\begin{enumerate}
	\item $\lim_{n\to \infty} \PR{\CRLC^n(R_n^*-\epsilon)\text{ satisfies }\cP} = 1$
	\item $\lim_{n\to \infty} \PR{\CRLC^n(R_n^*+\epsilon)\text{ satisfies }\cP} = 0$.
\end{enumerate}  

For the first statement, let $C = \CRLC^n(R_n^*-\epsilon)$. For each $\tau \in T_n$,  Lemma~\ref{lem:tau-threshold}\ref{enum:SharpnessForRLCLowRate} guarantees that $\PR{C\text{ contains }\tau} \le q^{-\epsilon n}$. We take a union bound over all $\tau \in T_n$ noting that $$|T_n|\le |\cD_{n,\ell}| \le \binom{n+q^\ell-1}{q^\ell-1}\le (n+q^\ell)^{q^\ell}$$
due to (\ref{eq:sizeD}). This yields
$$\PR{C\text{ satisfies }P_n} \le (n+q^\ell)^{q^\ell}\cdot q^{-\eps n} \le o_{n\to \infty}(1).$$

We turn to the second statement. Let $C = \CRLC^n(R_n^*+\epsilon)$, and let $\tau\in T_n$ such that 
\[ \max_{\tau'\in \cI_\tau}\TRLCEXP(\cP_\tau') = R^*.\] 
By Lemma~\ref{lem:tau-threshold}\ref{enum:SharpnessForRLCHighRate}, $C$ almost surely contains $\tau$, which is a sufficient condition for the code not to satisfy $\cP$.
\end{proof}
\begin{remark}[Probability of satisfying $\cP$ in Theorem \ref{thm:GeneralRLCThreshold}]\label{rmk:RLCThersholdProb}
	Fix $\epsilon > 0$. An inspection of the proof of Theorem \ref{thm:GeneralRLCThreshold} shows that $\CRLC^n(\TRLC^n(\cP)-\epsilon)$ satisfies $\cP$ with probability $1-2^{-\Omega(n)}$. Likewise, $\CRLC^n(\TRLC^n(\cP)+\epsilon)$ satisfies $\cP$ with probability $2^{-\Omega(n)}$.
\end{remark}

\begin{remark}[Relationship to random graphs]\label{rem:relationship_to_graphs} Lemma~\ref{lem:tau-threshold} has an analog in the theory of random graphs. Fix a constant-sized graph $H$ and let $G$ be a random graph in the $G(n,p)$ model. A natural problem is to determine the threshold for the appearance of $H$ as a sub-graph of $G$. The answer (see for example \cite[Sec.~4.2]{Bollobas11}) is that a copy of $H$ is likely to occur in $G$ whenever $p$ is large enough so that every subgraph of $H$ has, in expectation, $\omega(1)$ copies as subgraphs of $G$. To complete the analogy, equate $H$ with $\tau$, and a subgraph of $H$ with a $\tau$-implied distribution. 

We also mention the recent breakthrough result of Frankston et al., which studies this relationship between thresholds and expectations of sub-structures in a more general framework \cite{FrankstonKNP19}. However, since the properties that they study are not necessarily local, it is impossible for that work to precisely pinpoint the thresholds, as we do in our work. 
\end{remark}

\subsection{Probability that a matrix is contained in a random $s$-LDPC code}
The second building block shows that given a matrix $M\in \F_q^{n\times \ell}$, the probability that $M$ is contained in a random $s$-LDPC code is not much larger than that of appearing in a random linear code, provided that $M$ is \emph{$\delta$-smooth} (defined below).

\begin{definition}[Smooth distribution]\label{def:smooth}
Let $\delta > 0$. We say that a distribution $\tau$ over $\F_q^\ell$ is  \deffont{$\delta$-smooth} if $ \Pr_{v\sim \tau}[\langle u,v \rangle \neq 0] \geq \delta$
for all $u \in \F_q^\ell \setminus \{0\}$.
 If $M \in \F_q^{n\times \ell}$ is such that $\tau_M$ is $\delta$-smooth, we also say that $M$ is \deffont{$\delta$-smooth}.
\end{definition}

Intuitively, a distribution $\tau$ is smooth if for any fixed codimension 1 subspace $W = \{x \in \F_q^\ell\suchthat \langle x,u\rangle = 0\}$, a sample from $\tau$ is never too likely to lie $W$. 

\begin{remark}[Relationship to distance]\label{remark:distance}
In coding-theoretic terms, $\tau_M$ is $\delta$-smooth if and only if the code $\inset{M u \suchthat u \in \F_q^\ell}$ has relative distance at least $\delta$ and $M$ is full-rank.  Indeed, the relative weight of any codeword $M  u$ in this code is  
\[ \frac{1}{n} \sum_{i \in [n]} \ind{ \langle Mu, e_i  \rangle \neq 0} =
\frac{1}{n} \sum_{i \in [n]} \ind{ \langle u, M^Te_i  \rangle \neq 0} = \Pr_{v \sim \tau_M}[\langle u,v \rangle \neq 0 ], \]
where $e_i \in \F_q^n$ denotes the $i$-th standard basis vector, i.e., the vector with $1$ in the $i$-th coordinate and $0$ elsewhere. Furthermore, note that if $M$ is not full-rank and $u$ is a non-zero vector in $\ker(M)$, then the left-hand side of the above is $0$. Hence, we require that $M$ be full-rank.
\end{remark}

The following lemma bounds the probability that a matrix with smooth row distribution is contained in a random LDPC code with sufficiently large sparsity parameter.
We prove this lemma in Section~\ref{sec:random_ldpc_matrix}. 

\begin{restatable}[Probability that a random LDPC code contains a matrix]{lemma}{ProbabilityOfMatrixInLDPC}\label{lem:ProbabilityOfMatrixInLDPC}
	For any $\delta, \epsilon >0$, prime power $q$, and $\ell \geq 1$ there exists $s_0 = s_0(\eps, \delta, q, \ell) \geq 1$ such that the following holds for any odd $s \geq s_0$, and sufficiently large $n$. 
	Let $M \in \F_q^{n \times \ell}$ be $\delta$-smooth.
	Then the probability $p$ that $M$ is contained in a random $s$-LDPC code of length $n$ and rate $R$ satisfies $$ p \le q^{- (1-\epsilon) \cdot (1-R)\cdot \ell \cdot n}.$$
\end{restatable}


Given a smooth distribution $\tau$, in light of Fact \ref{fact:random_linear_matrix}, Lemma~\ref{lem:ProbabilityOfMatrixInLDPC} says that the expected number of matrices from $\cM_{n,\tau}$ in a random $s$-LDPC code is not much larger than this number for a random linear code. If we ignore the constraint that $\tau$ must be smooth, then together with Lemma~\ref{lem:tau-threshold} the above would imply Theorem \ref{thm:main}. Indeed, if a distribution $\tau$ is unlikely to appear in a random linear code then Lemma~\ref{lem:tau-threshold} shows that some $\tau$-implied distribution $\tau'$ appears $o(1)$ times in expectation in the random linear code. By Lemma~\ref{lem:ProbabilityOfMatrixInLDPC}, $\tau'$ appears $o(1)$ times in the random LDPC code as well, so the LDPC code is unlikely to contain $\tau'$. Thus, it is also unlikely to contain $\tau$.  (Of course, we cannot ignore the constraint that $\tau$ must be smooth; we will address this in our next building block discussed in Section~\ref{sec:hl_distance}).

The proof of Lemma~\ref{lem:ProbabilityOfMatrixInLDPC} proceeds by Fourier analysis.  
The basic idea is as follows: since $C$ is a random $s$-LDPC code, each parity-check corresponds (essentially) to an independent and uniformly random set of $s$ coordinates in $[n]$.\footnote{This is not exactly true because the parity checks that belong to the same layer are not independent; however, we show that this does not significantly affect the probability of the event of interest.
} Thus, the probability that a matrix $M \in \cM_{n, \tau}$ is in $C$ can be derived from the probability that $s$ random vectors $v_1, \ldots, v_s \sim \tau$ sum to zero. This probability is given by a convolution $\tau^{*s}(0) = \tau * \tau * \cdots * \tau(0)$ of $\tau$ with itself $s$ times.  The convolution is in turn controlled by $s$'th powers of the Fourier coefficients $\hat{\tau}(w)$ of $\tau$.  As we will see, the condition that $\tau$ be $\delta$-smooth implies that the nonzero Fourier coefficients $\hat{\tau}(w)$ are bounded away from $1$, and this means that if $s$ is large enough, the contributions $\hat{\tau}(w)^s$ of the nonzero coefficients to $\tau^{* s}(0)$ will become small.

\subsection{Distance of random $s$-LDPC codes}\label{sec:hl_distance}
As noted above, the first two building blocks show that for any $\delta$-smooth distribution $\tau \sim \F_q^\ell$, a random LDPC code of rate slightly below $\TRLC^n(\cP_\tau)$ is unlikely to contain $\tau$. The third and final building block shows that we may restrict our attention to $\delta$-smooth distributions.

As noted in Remark~\ref{remark:distance}, the condition that $M$ be $\delta$-smooth is the same as the condition that the code generated by $M$ has relative distance at least $\delta$. Thus, if $C \subset \F_q^n$ has relative distance at least $\delta$, it does not contain any matrices that are not $\delta$-smooth. Fortunately, it is well-known that \em binary \em random $s$-LDPC codes have good distance, and that in fact the distance approaches the \em Gilbert-Varshamov \em (GV) bound with high probability. Theorem~\ref{thm:ldpc_distance} generalizes this result to $s$-LDPC codes over any alphabet. Below, $h_q(x)$ is the $q$-ary entropy function (as in \eqref{eq:qary_entropy}).

\begin{restatable}[Random LDPC codes achieve the GV bound]{theorem}{LDPCDistance}\label{thm:ldpc_distance}
For any $\delta \in (0,1-1/q)$, $\epsilon>0$,  and prime power $q$ there exists $s_0 = s_0(\eps,\delta,q) \geq 1$ such that the following holds for any $s \geq s_0$. Let $R \leq 1-h_q(\delta)-\epsilon$. Then a random $s$-LDPC code of rate $R$ over $\F_q$ has relative distance at least $\delta$ with high probability.
\end{restatable}

\begin{remark}[Comparison to Gallager's proof]
Gallager's proof for binary random $s$-LDPC codes in~\cite{Gal62} uses generating functions.  We give an alternative proof using ideas from exponential families, which follows the approach of recent work by Linial and the first author~\cite{LM18}.  Our proof extends to random $s$-LDPC codes over any alphabet.  We note that Gallager left it as an open problem in~\cite{Gal62} to obtain a result like this for larger alphabets, but his definition was slightly different than ours: the coefficients $\alpha_{i,j}$ in his parity checks were all $1$'s, while ours are taken randomly from $\F_q^*$. 

Despite having different frameworks, our proof and that of \cite{Gal62} turn out to yield similar equations. In particular our proof of Lemma~\ref{lem:phifacts} is very similar to the corresponding proof in \cite{Gal62} at a technical level.  We highlight where the proofs diverge in Remark~\ref{rem:Gal-v-us}.
\end{remark}

\subsection{Proof of Theorem~\ref{thm:main} from Lemma~\ref{lem:tau-threshold}, Lemma~\ref{lem:ProbabilityOfMatrixInLDPC} and Theorem~\ref{thm:ldpc_distance}}
Theorem \ref{thm:main} now follows as an immediate consequence of the building blocks above.
We restate Theorem~\ref{thm:main} here:

\mainTheorem*

\begin{proof}
Fix a sufficiently large odd integer $s$ (depending on $\bar R$, $\epsilon$, $q$ and $b$). For  $n\in \N$, let $C:=\CLDPC{s}^n(R_n)$ for some $R_n\le \TRLC^n(\cP)-\epsilon$. Let $T_n$ be as in Observation \ref{obs:LocalPropertyDecomposition}. Let
\[
	\delta := \frac{h_q^{-1}(1 - \bar R)}2 > 0. 
\]

Fix some $\tau\in T_n$. Let $\tau'\in \cI_\tau$ be a maximizer of $\TRLCEXP(\tau')$. We may assume that $\tau'$ is a distribution over $\F_q^{d(\tau')}$, where we recall that $d(\tau') = \dim(\mathrm{span}(\supp(\tau')))$. Indeed, otherwise, let $A:\spn(\supp(\tau'))\to \F_q^{d(\tau')}$ be a linear bijection, and take the distribution of $Au$ (for $u\sim \tau'$) in place of $\tau'$ itself.

By Lemma \ref{lem:tau-threshold}, for $n$ large enough, 
\begin{align*}
R_n &\le \TRLC^n(\cP)-\epsilon \\
& \le \TRLC^n(\cP_\tau)-\epsilon \\
&\le \TRLCEXP(\tau')-\frac \epsilon 2 \\
&= 1-\frac{H_q(\tau')}{d(\tau')}-\frac \epsilon 2,
\end{align*}
where the first line is our assumption on $R_n$; the second line follows from the fact that any code satisfying $\mathcal{P}$ must in particular satisfy $\cP_\tau$; the third line is Lemma~\ref{lem:tau-threshold}; and the fourth line is the definition of $\TRLCEXP(\tau')$.


Consider the case where $\tau'$ is $\delta$-smooth. Let $p$ denote the probability that $C$ contains a given matrix from $M_{n,\tau'}$. By Lemma \ref{lem:ProbabilityOfMatrixInLDPC}, for $s$ large enough we have $p \le q^{-(1-\frac\epsilon 4)(1-R_n)\cdot d(\tau')\cdot n}$. Thus, the expected number of such matrices in $C$ is at most
\begin{align*}
|M_{n,\tau'}|\cdot p&\le q^{H(\tau')\cdot n}\cdot p \\
& \le q^{\left(H(\tau')-\left(1-\frac \epsilon 4\right)(1-R_n)\cdot d(\tau')\right)\cdot n}\\ 
&\le q^{\left(H(\tau')-(1-\frac \epsilon 4)\left(\frac {H(\tau')}{d(\tau')}+\frac \epsilon 2\right)\cdot d(\tau')\right)\cdot n}\\
&= q^{\left(\frac \epsilon 4\cdot H(\tau')-(1-\frac \epsilon 4)\frac \epsilon 2\cdot d(\tau')\right)\cdot n}\\
&\le q^{\left(\frac \epsilon 4\cdot d(\tau')-(1-\frac \epsilon 4)\frac \epsilon 2\cdot d(\tau')\right)\cdot n}\\
&\le q^{-\frac\epsilon 8\cdot d(\tau')\cdot n} \\
&\le q^{-\frac{\epsilon}8n}.\numberthis \label{eq:MainProofSmoothInCode}
\end{align*}
Here, we used the fact that $H(\tau') \le \log_q |\supp(\tau')| = d(\tau')$.

On the other hand, assume that $\tau'$ is not $\delta$-smooth. Let $D$ denote the event that the relative distance of $C$ is less than $\delta$. By Remark \ref{remark:distance}, if $C$ contains $\tau'$ then the event $D$ must hold (in the setting of that remark, our assumption that the domain of $\tau'$ is $\F_q^{d(\tau')}$, is equivalent to $M$ having full-rank). Since any code containing $\tau$ must also contain $\tau'$,
\begin{align*}
		\PR{C \text{ contains }\tau,\text{ and }D\text{ does not hold}} &\le \PR{C \text{ contains }\tau',\text{ and }D\text{ does not hold}} \\
		&= \PR{\exists M \in M_{n,\tau'} \text{ s.t. } M \in C,\text{ and }D\text{ does not hold}} \\
		&\le \PR{\exists M \in M_{n,\tau'} \text{ s.t. } M \in C} \le q^{-\frac{\epsilon}8 n},
\end{align*}
where the last inequality applies Markov's inequality and (\ref{eq:MainProofSmoothInCode}). Taking a union bound over all $\tau \in T_n$ and using (\ref{eq:sizeD}), we get
$$\Pr(C\text{ satisfies }\cP,\text{ and }D\text{ does not hold}) \le q^{-\frac{\epsilon}8 n}\cdot |T_n|\le q^{-\frac{\epsilon}8 n}\cdot |\cD_{n,\ell}|\le q^{-\frac\epsilon 8 n}\cdot \binom{n+q^\ell-1}{q^\ell-1}\le o_{n\to \infty}(1).$$

Finally, for $s$ large enough Theorem \ref{thm:ldpc_distance} says that $D$ almost surely does not hold. Thus, we conclude that $\C$ satisfies $\cP$ with high probability.
\end{proof}

%% file: RandomLinearMatrix.tex
In this section we prove Lemma~\ref{lem:tau-threshold}, which we restate below.

\sharpnessRandomLinear*

We note that statements \ref{enum:SharpnessForRLCLowRate} and \ref{enum:SharpnessForRLCHighRate} of Lemma~\ref{lem:tau-threshold} also imply the rest of the lemma. Thus, it suffices to prove them.

\subsection{Proof of Statement \ref{enum:SharpnessForRLCLowRate}}
Assume that $\tau$ is such that $R_\tau^*=\max_{\tau'\in \cI_\tau} \TRLCEXP(\tau')$ satisfies
\[
	R \leq R_\tau^* - \eps \ .
\]

Choose $\tau' \in \cI_\tau$ achieving $\TRLCEXP(\tau')=R_\tau^*$ and let $A \in \F_q^{m \times \ell}$ be such that $\tau'$ is given by $Av$ for $v \sim \tau$. 
By Fact~\ref{fact:random_linear_matrix}, a matrix $M' \in \cM_{n,\tau'}$ is contained in $C = \CRLC^n(R)$ with probability $q^{-(1-R) \cdot \rank(M') \cdot   n}=q^{- (1-R) \cdot d(\tau') \cdot n}$, and so
 $$\PR{ \exists M \in \cM_{n,\tau'}, M \subset C} \leq |\cM_{n,\tau'}| \cdot q^{-(1-R) \cdot d(\tau') \cdot n} \leq q^{(H_q(\tau')-(1-R) \cdot  d(\tau')) \cdot n} \leq q^{-\epsilon n},$$
where the first inequality follows by a union bound, the second applies  Fact~\ref{fact:orbit-size}, and the final inequality uses $\TRLCEXP(\tau') = 1-\frac{H_q(\tau')}{d(\tau')} \geq R+\eps$.

Finally, note that if $C$ contains
 some matrix $M \in \cM_{n,\tau}$, then by linearity, $M':=M  A^T \in \cM_{n,\tau'}$ is also contained in $C$. So we conclude 
\[ \PR{ \exists M \in \cM_{n,\tau}, M \subset C }  \leq q^{-\eps n}. \]

\subsection{Proof of Statement \ref{enum:SharpnessForRLCHighRate}}
We now proceed to the second part of the theorem, which is more involved.
Suppose that $\tau \in \cD_{n,\ell}$ is such that $R_\tau^* = \max_{\tau'\in \cI_\tau}\TRLCEXP(\tau')$ satisfies $R \geq R_\tau^* + \eps$.


First, we will argue that we may assume without loss of generality that $d(\tau) = \ell$. For if $d(\tau) < \ell$, by the definition of $d(\tau)$, there is some matrix $B \in \F_q^{d(\tau) \times \ell}$ of rank $d(\tau)$ so that the distribution $\tilde{\tau}$ given by $Bv, v \sim \tau$ has $d(\tilde{\tau}) = d(\tau)$. Note that $\tilde{\tau}$ is defined over $\F_q^{d(\tau)}$ and furthermore that $d(\tilde \tau) = d(\tau)$. We claim that 
\[  \max_{\tau'\in \cI_\tau}\TRLCEXP(\tau') \leq R-\eps \]
implies that 
\[  \max_{\tilde{\tau}'\in \cI_{\tilde{\tau}}}\TRLCEXP(\tilde{\tau}') \leq R-\eps. \]
To see this, we prove the contrapositive. Suppose that there is some $\tilde{\tau}' \in \cI_{\tilde{\tau}}$ so that $\TRLCEXP(\tilde{\tau}') > R - \eps$.
Then by the definition of $\cI_{\tilde{\tau}}$, there is some matrix $A \in \F_q^{m \times d(\tilde{\tau})}$ where $m \leq d(\tilde \tau)$ so that $\tilde{\tau}'$ is given by $Aw$, $w \sim \tilde{\tau}$. But this is the same as the distribution $ABv$, $v \sim \tau$, using the definition of $\tilde{\tau}$. Thus, $\tilde{\tau}' \in \cI_{\tau}$, and this implies that $\max_{\tau' \in \cI_\tau} \TRLCEXP(\tau') > R - \eps$. This establishes the contrapositive of the implication we wished to prove. 
Finally, we observe that ${n + q^{2\ell} - 1 \choose q^{2\ell} - 1 }$ is increasing in $\ell$. Therefore to prove the statement \ref{enum:SharpnessForRLCHighRate}, we may as well work with the distribution $\tilde{\tau}$ on $\F_q^{d(\tilde{\tau})}$. Indeed, if we can show 
\[
	\PR{\exists \tilde M \in \cM_{n,\tilde \tau},\tilde M \subseteq C} \geq 1 - \binom{n+q^{2d(\tau)}-1}{q^{2d(\tau)}-1} \cdot q^{-\eps n}
\]
then we obtain statement~\ref{enum:SharpnessForRLCHighRate} as
\begin{align*}
	\PR{\exists M \in \cM_{n,\tau}, M \subseteq C} &\geq \PR{\exists \tilde M \in \cM_{n,\tilde \tau},\tilde M \subseteq C}\\ 
	&\geq 1 - \binom{n+q^{2d(\tau)}-1}{q^{2d(\tau)}-1} \cdot q^{-\eps n} \geq 1 - \binom{n+q^{2\ell}-1}{q^{2\ell}-1} \cdot q^{-\eps n}.
\end{align*}
Thus, by replacing $\tau$ by $\tilde \tau$ and redefining $\ell = d(\tau) = d(\tilde \tau)$, we may assume in the following that $d(\tau) = \ell$.

For a matrix $M \in \F_q^{n \times \ell}$, let $X_M$ be the indicator variable for the event that $M \subseteq C$, and let $X = \sum_{M \in \cM_{n,\tau}}X_M$. Our goal then is to show that $X>0$ with high probability, 
and we do so by showing that $\Var{X} = o(\EE^2[X])$. 

 We first show a lower bound on $\EE[X]$. By Facts \ref{fact:random_linear_matrix} and \ref{fact:orbit-size}, 	\begin{equation}\label{eq:FirstMoment}
	\E {X}=  |\cM_{n,\tau}| \cdot  q^{-(1-R) \cdot \ell \cdot n} \geq q^{(H_q(\tau)-(1-R)\cdot \ell)\cdot n}\cdot  {n+q^\ell -1 \choose q^\ell -1}^{-1}.
	\end{equation}
	
	Next we show an upper bound on $\Var{X}$. Given a pair of matrices $M,M' \in \cM_{n,\tau}$, we let $(M|M')$ denote the 
	 $(n\times (2\ell))$-matrix consisting of a left $n\times \ell$ block equal to $M$, and a right $n \times \ell$ block equal to $M'$. Then in this notation we have
	\begin{eqnarray*}
	\Var{X}  & = &\sum_{M,M' \in \cM_{n,\tau}}  \bigg(\EE[X_M \cdot X_{M'}] - \EE[X_M] \cdot \EE[X_{M'}] \bigg)\\
	& = & \sum_{M,M' \in\cM_{n,\tau}}  \bigg(\Pr\left[(M|M') \subseteq C\right] -\Pr[M \subseteq C]\cdot\Pr[M' \subseteq C]\bigg)	\\
	& = & \sum_{M,M' \in \cM_{n,\tau}} \bigg(q^{-(1-R)\cdot\rank(M|M')\cdot n} - q^{-2 \cdot (1-R)\cdot \ell\cdot n} \bigg).
	\end{eqnarray*}

	Notice that in the above sum, terms for which $\rank(M|M')=2\ell$ vanish. Let
	$$\cM:=\bigg\{ (M|M') \mid M,M' \in \cM_{n,\tau}\; \text{and}\; \rank(M|M')<2\ell\bigg\},$$
	and 
\begin{equation}\label{eq:defD}
\cD:=\{\tau_M \suchthat M \in \cM \}.
\end{equation}
 Then we have
	\begin{eqnarray*}
\Var{X}   &\leq & \sum_{M \in \cM} q^{-(1-R)\cdot\rank(M)\cdot n}  \\
	& = & \sum_{\tau' \in \cD} \sum_{M \in \cM_{n,\tau'}} q^{-(1-R)\cdot \rank(M) \cdot n} \\
	& = & \sum_{\tau' \in \cD} |\cM_{n,\tau'}| \cdot q^{-(1-R)\cdot d(\tau') \cdot n}.\\
	& \leq & \sum_{\tau' \in \cD} q^{(H_q(\tau')-(1-R)\cdot d(\tau') )\cdot n}
	\end{eqnarray*}
where the last inequality follows by Fact~\ref{fact:orbit-size}.
	Finally, Claim~\ref{clm:joint-ent} below shows that for any $\tau' \in \cD$,
	$$H_q(\tau')-(1-R)\cdot d(\tau')\leq 2(H_q(\tau) - (1-R )\cdot \ell)-\epsilon,$$
	which implies in turn that
	\begin{equation}\label{eq:SecondMoment}
	\Var{X} \leq |\cD| \cdot  q^{2(H_q(\tau)-(1-R)\cdot \ell)\cdot n} \cdot q^{-\epsilon n} \leq {n+q^{2\ell} -1 \choose q^{2\ell} -1} \cdot  q^{2(H_q(\tau)-(1-R)\cdot \ell)\cdot n} \cdot q^{-\epsilon n} .
	\end{equation}
Above, we used the fact that $\cD \subseteq \cD_{n, 2\ell}$ and applied \eqref{eq:sizeD}.	
	Combining (\ref{eq:FirstMoment}) and (\ref{eq:SecondMoment}), by Chebyshev's inequality we conclude that
	$$\Pr[X =0 ] \leq  \frac  {\Var{X}} {\EE^2[X]} \leq {n+q^{2\ell} -1 \choose q^{2\ell} -1}^3 \cdot q^{-\epsilon n}  .$$
To complete the proof, we prove Claim~\ref{clm:joint-ent} which we used above.	
	\begin{claim}\label{clm:joint-ent}  Let $\cD$ be as in \eqref{eq:defD}.
	For any $\tau' \in \cD$,
	\begin{equation*}\label{eq:var}
	H_q(\tau')-(1-R)\cdot d(\tau')\leq 2(H_q(\tau) - (1-R) \cdot \ell)-\eps.
	\end{equation*}
	\end{claim}
	
	\begin{proof}
	In what follows, let $d:=d(\tau')$, and $V:=\spn(\supp(\tau')) \subseteq \F_q^{2\ell}$. 
Let $w_1, \ldots, w_{2\ell - d} \in \F_q^{2\ell}$ be a basis for $V^\perp$.
Let $\pi_1: \F_q^{2\ell} \to \F_q^\ell$ (respectively, $\pi_2$) denote the projection of a vector $w \in \F_q^{2\ell}$ to the first (respectively, last) $\ell$ coordinates. We also apply $\pi_1$ and $\pi_2$ to subsets $X \subseteq \F_q^{2\ell}$, defining $\pi_1(X) := \{\pi_1(x):x \in X\}$. In particular, note that as $\tau' \in \cD$, it follows that $\pi_1(\supp(\tau')) = \pi_2(\supp(\tau')) = \supp(\tau)$.

Finally, let $A$ be the matrix whose rows are $w_1, \ldots, w_{2\ell-d}$, and let
$A_1 \in \F_q^{(2\ell-d) \times \ell}$ (respectively, $A_2$) denote the matrix whose rows are $\pi_1(w_1), \ldots, \pi_1(w_{2\ell-d})$  (respectively, $\pi_2(w_1), \ldots, \pi_2(w_{2\ell-d})$). That is,
\begin{align*}
	A  = \left[\begin{array}{@{} c c c@{}}
		 & w_1 & \\ & w_2 & \\ & \vdots & \\ & w_{2\ell-d} & \\
	\end{array}\right] = \left[
		\begin{array} {@{} c c c | c c c@{}}
			& \pi_1(w_1)& & & \pi_2(w_1) & \\ 
			& \pi_1(w_2) & & &  \pi_2(w_2) & \\ 
			& \vdots & & &  \vdots \\
			& \pi_1(w_{2\ell-d})& & & \pi_2(w_{2\ell-d}) & \\
		\end{array}
	\right] = \left[\begin{array} {@{} c c c | c c c@{}}
		& & & & & \\ 
		& A_1 & & & A_2 & \\ 
		& & & & & \\
		\end{array}\right].
\end{align*}



We claim that all rows of $A_1$ are linearly independent, and so $\rank(A_1) = 2\ell-d$. To see this suppose in contradiction that $\pi_1(w_1), \ldots, \pi_1(w_{2\ell-d})$ are linearly dependent. Then there exists a non-trivial linear combination of $w_1, \ldots, w_{2\ell - d}$ that sums to a non-zero vector of the form $(0,w)$. But this means that $\pi_2(\supp(\tau')) = \supp(\tau)$ is orthogonal to $w$, in contradiction to our assumption that $\spn(\supp(\tau)) = \F_q^\ell$. Consequently, recalling that $d(\tau)=\ell$, the distribution $\tau''$ given by $A_1  w$ for $w \sim \tau$ has $d(\tau'')=2\ell-d$. As $\tau'' \in \cI_\tau$, $\TRLCEXP(\tau'') \leq R-\eps$.

Let $I_q(X;Y) = H_q(X)-H_q(X\mid Y)$ denote the \deffont{base-$q$ mutual information} of $X$ and $Y$.
	Now for $v \sim \tau'$ we have,
	\begin{align}
	H_q(\tau') & = H_q(v) \nonumber \\  
	& =  H_q(\pi_1(v)) + H_q(\pi_2(v)) - I_q(\pi_1(v);\pi_2(v)) \label{eq:mutual-info}\\
	&=  2H_q(\tau) - I_q(\pi_1(v);\pi_2(v)) \label{eq:injectivity} \\
	& \leq 2H_q(\tau) - I_q(A_1  \pi_1(v);-A_2  \pi_2(v)) \label{eq:data-processing}\\
	&= 2H_q(\tau) - H_q(A_1 \pi_1(v) ) \label{eq:Av} \\
	& \leq 2H_q(\tau) - (1-R+\epsilon) \cdot d(\tau'') \label{eq:rearranging}\\
	& = 2H_q(\tau) - (1-R+\epsilon) \cdot (2\ell-d). \nonumber
	\end{align}
	The equality \eqref{eq:mutual-info} follows from the definition of mutual information, using $v = (\pi_1(v),\pi_2(v))$. The equality \eqref{eq:injectivity} follows from the fact that $\pi_1$ and $\pi_2$ are injective on the row-span of $A$. The inequality \eqref{eq:data-processing} follows from the data-processing inequality. The equality \eqref{eq:Av} follows since $A_1\pi_1(v) + A_2\pi_2(v) = Av=0$. Finally, inequality \eqref{eq:rearranging} follows because $1-\frac{H_q(\tau'')}{d(\tau'')} = \TRLCEXP(\tau'') \leq R-\eps$. 
	Rearranging, and recalling the assumption that $2\ell > d$, gives the desired conclusion.
	\end{proof}

%% file: RandomLDPCMatrix.tex
In this section we prove our second building block, Lemma~\ref{lem:ProbabilityOfMatrixInLDPC}, which we re-state below. For the reader's convenience, we recall that a distribution $\tau \sim \F_q^\ell$ is said to be $\delta$-smooth (for some $\delta>0$) if $\Pr_{v\sim \tau}[\langle u,v \rangle \neq 0] \geq \delta$ for all $u \in \F_q^\ell \setminus \{0\}$.


\ProbabilityOfMatrixInLDPC*

\begin{remark}[The parity of $s$, again] Lemma~\ref{lem:ProbabilityOfMatrixInLDPC} holds for even $s$ as well as odd $s$, but the proof is slightly simpler for odd $s$, so we state and prove it in this case for clarity.  This is the only place in the proof of Theorem~\ref{thm:main} where we use the parity of $s$, and so this remark implies Remark~\ref{rem:parity}. 
\end{remark}

We begin with some definitions from Fourier analysis which we will need.
\subsection{Fourier-analytic facts}
We give here some basic definitions and facts from Fourier analysis of functions on $\F_q$.  We refer the reader to, for example, \cite{LNtext,Ryantext}
for more details and proofs of these facts. 
In what follows assume that $q = p^h$ for a prime $p$.  The \deffont{trace map} of $\F_q$ over $\F_p$ is the function $\tr:\F_q \to \F_p$  given by
\[ \tr(\alpha) = \alpha + \alpha^p + \alpha^{p^2} + \cdots + \alpha^{p^{h-1}}. \]
 For a function $f: \F_q^n \to \CC$, we define the \deffont{Fourier transform} $\hat{f}: \F_q^n \to \CC$ of $f$ by
\[ \hat{f}(y) = \Eover{x \in \F_q^n} {f(x) \cdot \overline{\chi_x(y)}},\]
where $y \in \F_q^n$, 
$\chi_x(y) = \omega_p^{\tr(\ip{x}{y})} $, and $\omega_p = e^{2\pi i / p}$. Then we have the decomposition
\[f (x) = \sum_{y\in \F_q^n} \hat f(y) \cdot \chi_y(x) \ .\]

We define an inner product on the space of $\CC$-valued functions on $\F_q^n$ by 
\[
	\ip{f}{g} = \Eover{x \in \F_q^n}{f(x) \cdot \overline{g(x)}} .
\]	
 \deffont{Plancherel's identity} then asserts that
\[
	\ip{f}{g} = \sum_{x \in \F_q^n}\hat{f}(x) \cdot \overline{\hat{g}(x)}.
\]
An important special case is \deffont{Parseval's identity}:
\[
	\ip ff = \sum_{x \in \F_q^n}|\hat{f}(x)|^2 .
\]

The \deffont{convolution}  of a pair of functions $f,g:\F_q^n \to \CC$ is given by
\[
	(f*g)(x) = \Eover{y \in \F_q^n}{f(y) \cdot g(x-y)} .
\]
Convolution interacts nicely with the Fourier transform:
\[
	\widehat{f*g}(x) = \hat{f}(x) \cdot \hat{g}(x) .
\]
As a useful piece of notation, we define inductively $f^{*1}:=f$, and $f^{*s} = f^{*(s-1)}*f$ for an integer $s\geq 2$. 

Finally, we state the following claim and, for lack of a suitable reference, provide the proof (although this fact is certainly well-known; in particular, it is very similar in spirit to \cite[Proposition~1.26]{Ryantext}). It allows us to write the probability that a sum of i.i.d. random variables from $\F_q^\ell$ takes a certain value in terms of the convolution of its density function.

\begin{claim} \label{claim:sum-to-conv}
	Let $P \sim \F_q^\ell$ be a distribution. For any $y \in \F_q^\ell$ and $s \geq 1$, if $u_1,\dots,u_s \sim P$ are independent, 
	\[
		\PR{\sum_{i=1}^s u_i=y} = q^{\ell(s-1)}\cdot P^{*s}(y) \ .
	\]
\end{claim}

\begin{proof}
	By induction on $s$. The case $s=1$ is clear as $\PR{u_1=y}=P(y)=P^{*1}(y)$, so we now assume $s>1$. Let $u_1,\dots,u_s$ be independent samples from $P$. 
	\begin{align*}
		\PR{\sum_{i=1}^s u_i = y} &= \sum_{v \in \F_q^\ell}\PR{u_s=v} \cdot \PR{\sum_{i=1}^{s-1}u_i=y-v|u_s=v} \\
		&= \sum_{v \in \F_q^\ell}P(v) \cdot \left(q^{\ell(s-2)}\cdot P^{*(s-1)}(y-v)\right) \\
		&= q^{\ell(s-1)}\cdot \Eover{v \in \F_q^\ell}{P(v)\cdot P^{*(s-1)}(y-v)} \\ 
		&= q^{\ell(s-1)}\cdot P^{*s}(y) \ .
	\end{align*}
	The second equality applied the induction hypothesis. 
\end{proof}

\subsection{Proof of Lemma \ref{lem:ProbabilityOfMatrixInLDPC}}\label{subsec:ldpc_main_technical}

	Let $H \in \F_q^{((1-R)\cdot n) \times n}$ be the parity-check matrix of $C$ with layers $H_1, H_2, \ldots, H_{(1-R)\cdot s}$, as in Figure~\ref{fig:F}. Recall that each layer $H_i$ is an  independent sample from $F  D  \Pi$, where $F$ is also as in Figure~\ref{fig:F},
	 $\Pi \in \{0,1\}^{n \times n}$ is a random permutation matrix, and $D \in \F_q^{n \times n}$ is a diagonal matrix with diagonal entries that are independent and uniformly random in $\F_q^*$. 
	Let $\Lambda$ be a random matrix sampled according to the distribution $\Pi  M$. Then by independence of the layers,
	\begin{align}\label{eq:Lambda}
	 \Pr[M \subseteq C] & = \Pr[H  M=0] \nonumber \\ 
	&= \big(\Pr[H_1   M=0]\big)^{(1-R)\cdot s} \nonumber \\
	& = \big(\Pr[F  D   \Pi  M =0]\big)^{(1-R)\cdot s} \nonumber  \\
	& = \big(\Pr[F D  \Lambda=0]\big)^{(1-R) \cdot s}.
	\end{align}
	So it suffices to bound the probability that $F D \Lambda =0$.
	
Next, observe that each row in $\Lambda$ has the marginal distribution $\tau_M$. Indeed, for each $i \in [n]$, if $\pi:[n] \to [n]$ denotes the random permutation corresponding to $\Pi$, the probability that the $i$-th row of $\Lambda$ takes value $v \in \F_q^\ell$ is precisely the probability that $v=u_{\pi^{-1}(i)}$, and $\pi^{-1}(i)$ is a uniformly random element of $[n]$. 
Let $\Lambda' \in \F_q^{n \times \ell}$ be a random matrix in which each row is independently sampled according to $\tau_M$. 
We claim that
\begin{equation}\label{eq:Lambda'}
\Pr[F D \Lambda=0] \leq  O\inparen{n^{\frac{q^\ell-1}{2}}} \cdot \Pr[F D \Lambda'=0].
\end{equation}
To justify \eqref{eq:Lambda'}, note that the distribution of $\Lambda$ is identical to the distribution $\Lambda'$, conditioned on the event that $\Lambda'$ is in the support of $\Lambda$. In other words, the two distributions are identical conditioned on $\Lambda'$ having the same type as $M$. Using our notation, this event is succinctly desribed as $\tau_{\Lambda'} = \tau_M$. Thus, 
\begin{align*}
\PR{F D \Lambda=0} &= \PR{ F  D\Lambda' = 0 \mid \tau_{\Lambda'} = \tau_M } \\
&= \frac{ \PR{ F D \Lambda' = 0 \wedge \tau_{\Lambda'} = \tau_M } }{ \PR{ \tau_{\Lambda'} = \tau_M} }\\
&\leq \frac{ \PR{ F D\Lambda' = 0 } }{ \PR{ \tau_{\Lambda'} = \tau_M } }.
\end{align*}
Now we have
\begin{align*}
 \PR{ \tau_{\Lambda'} = \tau_M } &= { n \choose n\cdot \tau_{M}(v_1), \ldots, n\cdot \tau_{M}(v_{q^\ell})} \cdot \prod_{ v \in \F_q^\ell } \tau_M(v)^{ n\cdot \tau_M(v) }
\end{align*}
where $v_1, \ldots, v_{q^\ell}$ are the elements of $\F_q^\ell$. Noting that $\prod_{ v \in \F_q^\ell } \tau_M(v)^{ n\cdot \tau_M(v) }=q^{-nH_q(\tau_M)}$, \eqref{eq:Lambda'} follows from Fact~\ref{fact:orbit-size}.

Thus, it is enough to bound the probability that $F D \Lambda'=0$. Let $P$ denote the distribution given by $\lambda v$ for $v \sim \tau_M$ and uniformly random $\lambda \in \F_q^*$. 
Using Claim~\ref{claim:sum-to-conv}, we can express this probability as
	\begin{align}\label{eq:conv}
		\Pr\left[F D \Lambda' =0\right] = \inparen{\Pr_{u_1,\dots,u_s \sim P}\left[\sum_{i=1}^su_i=0\right]}^{n/s} = \bigg(q^{\ell \cdot (s-1)} \cdot P^{*s}(0)\bigg)^{n/s}.
	\end{align}
	Next we bound $P^{*s}(0)$.
	 In terms of Fourier transform, we can write
	\[
		P^{*s}(0) = \sum_{y \in \F_q^\ell}\widehat{P^{*s}}(y) \cdot \chi_y(0) = \sum_{y \in \F_q^\ell}\inparen{\hat{P}(y)}^s.
	\]	
Claim \ref{claim:fourier-bound} below shows that 
				$
			\hat{P}(y) \leq q^{-\ell} \cdot \inparen{1-\frac{q}{q-1}\cdot \delta }
		$ for any $y \in \F_q^\ell\setminus\{0\}$ (in particular, it's a real number),
		and by the assumption that $s$ is odd this implies in turn that 
	\begin{equation}\label{eq:fourier}
		P^{*s}(0) = \inparen{\hat{P}(0)}^s + \sum_{y \in \F_q^\ell \setminus \{0\}}\inparen{\hat{P}(y)}^s 
		\leq q^{-\ell \cdot s} + q^{-\ell\cdot (s-1)} \cdot \left(1-\frac{q}{q-1}\cdot\delta \right)^s.
	\end{equation}
	
	Finally, combining Equations (\ref{eq:Lambda}), (\ref{eq:Lambda'}), (\ref{eq:conv}), and (\ref{eq:fourier}) we conclude that
	$$\Pr[M\subseteq C] \leq   O\inparen{n^{\frac{q^\ell-1}{2}\cdot (1-R) \cdot s}} \cdot  \left(q^{-\ell } + \left(1-\frac{q}{q-1}\cdot\delta\right)^s\right)^{(1-R) \cdot n} \leq  q^{-(1-\epsilon) \cdot(1-R)\cdot \ell \cdot  n},$$
 where the last inequality holds for large enough $s$ depending on $\delta,\epsilon, q, \ell$, and sufficiently large $n$.

\begin{remark}[The choice of $s$]\label{rmk:s2}  
An inspection of the last line of the proof shows that we may take
\[ s_0 = O\inparen{ \frac{ \ell }{ \log_q\inparen{ \frac{ 1}{ 1 - \delta/(1 - 1/q) } } } }. \]
In particular, noting that $\ell \leq b$ and that 
\[ \log_q\inparen{ \frac{1}{1 - \delta/(1-1/q)} } = \frac{1}{\ln(q)} \sum_{i=1}^\infty \frac{1}{i} \inparen{ \frac{\delta}{1 - 1/q} }^i, \]
this part of the proof requires us to take 
\[ s_0 \geq C_0 \cdot \frac{ b \log(q) }{\delta } \]
for some constant $C_0 > 0$.
There is one other place in the proof of Theorem~\ref{thm:main} that requires $s_0$ to be sufficiently large; we comment on this again in Remark~\ref{rmk:s3}.
\end{remark}

Now, all that remains is to prove Claim~\ref{claim:fourier-bound} which we used above.
		
	\begin{claim} \label{claim:fourier-bound}
		For any $y \in \F_q^\ell \setminus\{0\}$, $\hat{P}(y) \in \R$ and 
				\[
			\hat{P}(y) \leq q^{-\ell} \cdot \inparen{1-\frac{q}{q-1} \cdot \delta}.
		\]
	\end{claim}
	
	\begin{proof} [Proof of Claim~\ref{claim:fourier-bound}]
		We have
		\begin{align*}
			\hat{P}(y) &= q^{-\ell}\cdot\sum_{x \in \F_q^\ell}P(x) \cdot \overline{\omega_p^{\tr(\langle y,x\rangle)}}\\ 
			&= q^{-\ell}\cdot\sum_{x \in \F_q^\ell}P(x) \cdot \omega_p^{-\tr(\langle y,x\rangle)}\\
			&= q^{-\ell} \cdot \Eover{x \sim P}{\omega_p^{-\tr(\langle y,x\rangle)}} \\
			&= q^{-\ell} \cdot \EE_{v \sim \tau_M} \Eover{\lambda \in \F_q^*}{\omega_p^{-\tr(\langle y,\lambda v\rangle)}} \\
			&= q^{-\ell}\cdot\inparen{ \Pr_{v \sim \tau_M} [\langle v,y \rangle \neq 0] \cdot  \EE_{\xi \in \F_q^*}\big[\omega_p^{\tr(\xi)}\big]+ \Pr_{v \sim \tau_M} [\langle v,y \rangle=0] \cdot \EE_{\lambda \in \F_q^*}
			\big[\omega_p^{\tr(0)}\big] }\\
			& =  q^{-\ell}\cdot\inparen{ \Pr_{v \sim \tau_M} [\langle v,y \rangle \neq 0] \cdot \frac{-1}{q-1}+ \Pr_{v \sim \tau_M} [\langle v,y \rangle=0] \cdot 1} \\
			& \leq q^{-\ell} \cdot \inparen{\frac{-\delta} {q-1} +(1-\delta)} 
			= q^{-\ell} \cdot \inparen{1-\frac{q}{q-1} \cdot \delta},
		\end{align*}
		where the last inequality follows by assumption that $\tau_M$ is $\delta$-smooth.
			\end{proof}
	
This completes the proof of Lemma~\ref{lem:ProbabilityOfMatrixInLDPC}.

%% file: distance.tex
In this section we prove Theorem~\ref{thm:ldpc_distance}, which shows that an LDPC code over any alphabet approaches the Gilbert-Varshamov bound with high probability.  We restate the theorem below.
\LDPCDistance*

\subsection{Proof of Theorem~\ref{thm:ldpc_distance}, given a lemma}
In this section we give an outline of the proof of Theorem~\ref{thm:ldpc_distance} and prove the theorem based on Lemma~\ref{lem:phifacts} that we state below and prove in subsequent subsections.

Our goal is to show that a random $s$-LDPC code $C$ has good distance, or equivalently that there are no low-weight codewords in $C$ with high probability.  To that end, we introduce the following notation.
\begin{definition}
For $\lambda \in (0,1)$ such that $\lambda n$ is an integer, let $P_\lambda = \Pr[u\in C]$, for $u\in \F_q^n$ with relative weight $\lambda$. Note that this probability is the same for every $u$ of weight $\lambda$, so $P_\lambda$ is well-defined.
\end{definition}
Our main challenge is to find sufficiently tight upper bounds on these terms $P_\lambda$ for $0<\lambda\le \delta$. 
The proof proceeds by giving a bound on $P_\lambda$ in terms of
a certain function $\varphi:(0,\frac{q-1}q]\to \R_{\le0}$.  
We will prove the following lemma below in Sections~\ref{subsec:phi} and \ref{subsec:phi2}.
We will define $\varphi$ below in Section~\ref{subsec:phi}, but for now we introduce its important properties in the following lemma (which we also prove below).

\begin{lemma}\label{lem:phifacts}
There is a function $\varphi: \left(0, \frac{q-1}{q} \right] \to \R_{\le 0}$ which has the following properties.
\begin{enumerate}
\item  \label{item:P_lambdaBound}
For every $\lambda \in \left(0,\frac {q-1}q\right]$, 
$$\log_q P_\lambda \le \varphi(\lambda)(1-R)n.$$
\item \label{item:phi_Bound} The function $\varphi$ satisfies
$$\varphi(\lambda)\le  \log_q\left(1+(q-1)\left(1-\frac{q}{q-1}\lambda\right)^s\right)-1$$
for all $\lambda \in (0,\frac {q-1}q]$. 
\item \label{item:P_lambdaMonotonicity}
The function $\frac{\varphi(\lambda)}{h_q(\lambda)}$ is strictly increasing in the range $0<\lambda \le \frac{q-1}q$.
\end{enumerate}
\end{lemma}

Before we prove Lemma~\ref{lem:phifacts}, we show how it implies Theorem~\ref{thm:ldpc_distance}.
\begin{proof}[Proof of Theorem \ref{thm:ldpc_distance}]
Our goal is to show that if $C$ is a random $s$-LDPC code as in the statement of Theorem~\ref{thm:ldpc_distance}, then with high probability there are no codewords in $C$ of relative weight less than $\delta$.
In the following, we assume without loss of generality that $\delta n$ is an integer.
Now
\begin{align} 
\Pr[C \text{ has relative distance less than }\delta] &\le \sum_{i=1}^{\delta n}P_{\frac in}\left|\left\{u\in \F_q^n\mid \wt{u}=\frac in\right\}\right| \label{eq:A}\\ 
&\le \sum_{i=1}^{\delta n}P_{\frac in}q^{nh_q(\frac in)} \notag \\
&\le \sum_{i=1}^{\delta n} q^{(\varphi(\frac in)(1-R)+h_q(\frac in))n} \label{eq:B}\\ 
&= \sum_{i=1}^{\delta n}q^{nh_q(\frac in)\left(\frac{(1-R)\varphi(\frac in)}{h_q(\frac in)}+1\right)} \\
&\le \sum_{i=1}^{\delta n}q^{nh_q(\frac in)\left(\frac{(1-R)\varphi(\delta)}{h_q(\delta)}+1\right)}.\label{eq:distanceFirstMoment}
\end{align}
Above, \eqref{eq:A} follows from the union bound, \eqref{eq:B} from Item \ref{item:P_lambdaBound} of Lemma~\ref{lem:phifacts}, and \eqref{eq:distanceFirstMoment} from Item \ref{item:P_lambdaMonotonicity} of Lemma~\ref{lem:phifacts}. By Item \ref{item:phi_Bound} of Lemma~\ref{lem:phifacts},
\[
	\frac{(1-R)\varphi(\delta)}{h_q(\delta)}+1 = \frac{(1-R)\cdot\inparen{\log_q\inparen{1+(q-1)\inparen{1-\frac{q}{q-1}\delta}^s}-1}}{h_q(\delta)} + 1.
\]
Recall our hypothesis that the rate of the code satisfies $R \leq 1 - h_q(\delta) - \eps$, and so $ 1-R \geq h_q(\delta)+\eps$. Noting that $\log_q\inparen{1+(q-1)\inparen{1-\frac{q}{q-1}\delta}^s}-1 \leq 0$ for any $\delta \in (0, 1 - 1/q)$ and for any $s \geq 1$,  we may thus bound the right hand side from above by
\begin{align*}
	&\frac{(h_q(\delta)+\eps)\cdot\inparen{\log_q\inparen{1+(q-1)\inparen{1-\frac{q}{q-1}\delta}^s}-1}}{h_q(\delta)} + 1 \\ 
	&= \inparen{1+\frac{\eps}{h_q(\delta)}}\cdot\inparen{\log_q\inparen{1+(q-1)\inparen{1-\frac{q}{q-1}\delta}^s}-1} + 1 \\
	&= \inparen{1+\frac{\eps}{h_q(\delta)}}\cdot\log_q\inparen{1+(q-1)\inparen{1-\frac{q}{q-1}\delta}^s} - \frac{\eps}{h_q(\delta)} \\
	&\leq \inparen{1+\frac{\eps}{h_q(\delta)}}\frac{(q-1)}{\ln(q)}\inparen{1-\frac{q\delta}{q-1}}^s - \frac{\eps}{h_q(\delta)} \\
	&\leq - \frac{\eps}{2h_q(\delta)},
\end{align*}
where the last inequality holds as long as $s$ is sufficiently large in terms of $\delta, \eps$ and $q$. Hence, we conclude that 
\[
	\frac{(1-R)\varphi(\delta)}{h_q(\delta)}+1 \leq - \frac{\eps}{2h_q(\delta)} \leq -\frac{\eps}{2} .
\]
Hence, the right-hand side of (\ref{eq:distanceFirstMoment}) is upper bounded by
$$\sum_{i=1}^{\delta n}q^{-\frac{nh_q(\frac in)\eps}{2}}.$$
This sum is dominated by its first term, so it is at most $O(n^{-\Omega(1)})$.

\end{proof}

\begin{remark}[The choice of $s$]\label{rmk:s3}
An inspection of the proof above shows that it suffices to take $s \geq C_1\cdot\ln(q/\eps)/\delta$ for some constant $C_1>0$.  Thus, this part of the proof requires that $s_0 \geq C_1\cdot\ln(q/\eps)/\delta$.
\end{remark}

\begin{remark}[Polynomially small failure probability]
In the proof, we see that the failure probability, while $o(1)$, is only polynomially small in $n$.  In fact, this is tight:
it is not hard to see that an $s$-random LDPC code  $C$ (for $s = O(1)$) contains a codeword of weight $2$ with probability $n^{-O(1)}$. 
\end{remark}

\subsection{The function $\varphi$ and proof of Lemma~\ref{lem:phifacts}, Items~\ref{item:P_lambdaBound} and \ref{item:phi_Bound}} \label{subsec:phi}
Let $\lambda \in \left(0,\frac {q-1}q\right]$ such that $\lambda n$ is an integer, and let $u\in \F_q^n$ have weight $\lambda n$. 
Let $H_1, \ldots, H_t$ be the layers of the the parity-check matrix $H$ of $\cC$, as in Figure~\ref{fig:F}.
Note that the matrices $H_1,\ldots, H_t$ are identically and independently distributed. In particular, the events $\Pr(H_iu=0)$ are independent. Hence,
\begin{equation}\label{eq:P_lambda}
P_\lambda = \Pr[u\in C] = \Pr[Hu=0] = \Pr[H_1u=0]^t.
\end{equation}
Since the distribution of $H_1$ is invariant to permutation of coordinates, this last probability does not depend on the vector $u$ as long as it is of relative weight $\lambda$. Hence, 
$$\Pr[H_1u=0] = \Pr[H_1\bar u=0] = \Pr[F\bar u=0],$$
where $\bar u$ is uniformly sampled from the set of all vectors of weight $\lambda$ in $\F_q^n$ (the last equality uses that $D\Pi \bar u$ is distributed identically to $\bar u$). Therefore,
$$P_\lambda = \Pr[F\bar u=0]^t,$$
where $F$ is as in Figure~\ref{fig:F}.

We turn to bound this expression. Let $\beta \in\left(0,\frac{q-1}q\right]$. Denote by $\mu_q(\beta)$ the distribution on $\F_q$ which is $0$ with probability $1-\beta$ and uniform on $\F_q^*$ with probability $\beta$. When $\beta$ is clear from context, we shorthand $\mu_q=\mu_q(\beta)$. Let $v\in \F_q^n$ be a random vector whose entries are i.i.d. random variables sampled according to $\mu_q$, which we denote by $v \sim \mu_q^n$. 
Observe that the distribution of $v$, conditioned on $\wt v=\lambda$, is identical to the distribution of $\bar u$. Indeed, for any fixed $x \in \F_q^n$ with $\wt x =\lambda$, we have 
\begin{align*}
	\PR{v=x|\wt v=\lambda} &= \frac{\PR{v=x \text{ and } \wt v=\lambda}}{\PR{\wt v = \lambda}} = \frac{\PR{v=x}}{\PR{\wt v = \lambda}} \\
	&= \frac{\inparen{\frac{\beta}{q-1}}^{\lambda n} \inparen{1-\beta}^{n-\lambda n}}{\binom{n}{\lambda n}\beta^{\lambda n}\inparen{1-\beta}^{n-\lambda n}} = \frac{1}{\inparen{q-1}^{\lambda n}\binom{n}{\lambda n}},
\end{align*}
that is, exactly 1 over the size of a Hamming ball of radius $\lambda$, which is $\PR{\bar u = \lambda}$. Hence, by Bayes' rule,
\begin{equation} \label{eq:FuToFv}
\Pr[F\bar u=0] = \Pr[Fv=0 \mid \wt v = \lambda] = \Pr[\wt v=\lambda \mid Fv=0]\cdot \frac{\Pr[Fv=0]}{\Pr[\wt v=\lambda]}\le \frac{\Pr[Fv=0]}{\Pr[\wt v=\lambda]}
\end{equation}
where the probabilities are over the choice of $v \sim \mu_q(\beta)^n$.

We proceed to bound the right-hand side of (\ref{eq:FuToFv}). For the denominator, note that
\begin{equation}\label{eq:P(wt(v)=lambda)}
\Pr[\wt v=\lambda] = \binom{n}{\lambda n} \beta^{\lambda n}(1-\beta)^{(1-\lambda) n} \ge q^{-\DKL{\lambda}{\beta}{q}n}
\end{equation}
where above $\DKL{x}{y}{q}$ denotes the KL Divergence,
\[ \DKL{x}{y}{q} = -x\log_q \frac yx - (1-x)\log_q\frac{1-y}{1-x} \text{ for } x\in [0,1] \text{ and } y\in (0,1). \]

We next focus on the numerator. The following notation will be useful:
\begin{definition}
For $k\in \N$, let
$$\V_q^k = \inset{w\in \F_q^k \suchthat \sum_{i=1}^k w_i=0}.$$
\end{definition}
Let $f_1,\ldots, f_{\frac ns}$ denote the rows of the matrix $F$. Note that the vectors $f_1,\ldots f_{\frac ns}$ have disjoint supports, so the products $f_i v$ are independently and identically distributed. Hence, $\Pr[Fv=0] = \Pr[f_1v=0]^{\frac ns}$. Observe that the distribution of $v$ is invariant under multiplication of each entry by a nonzero element of $\F_q$. Consequently,
\begin{equation}\label{eq:P(Fv=0)}
\Pr_{v \sim \mu_q^n}[Fv=0] = \Pr_{v \sim \mu_q^n}[f_1v=0]^{\frac ns} = \Pr_{v \sim \mu_q^n}\left[\sum_{i=1}^s v_i = 0\right]^{\frac ns} = \left( \Pr_{w\sim \mu_q^s}[w\in \V_q^s]\right)^{n/s} .
\end{equation}

The following lemma gives a closed form for this last expression.
\begin{lemma} \label{lem:SumVi=0}
$$\Pr_{w\sim \mu_q^s}[w\in \V_q^s] = \frac {1+(q-1)\left(1-\frac {q\beta}{q-1}\right)^s}q.$$
\end{lemma}
\begin{proof}
We proceed by induction.  The base case ($s=0$) is immediate.  Now suppose that the statement holds for $s-1$ and let $\pi: \F_q^s \to \F_q^{s-1}$ denote the projection onto the first $s-1$ coordinates.  Then
\begin{align*}
\Pr_{w \sim \mu_q^s} [w \in \V_q^s] &= \Pr_{w \sim \mu_q^s}\left[ \pi(w) \in \V_q^{s-1} \right] \cdot \Pr_{w \sim \mu_q^s}[w_s = 0] + \Pr_{w \sim \mu_q^s}\left[\pi(w) \not\in \V_q^{s-1}\right] \cdot \Pr_{w \sim \mu_q^s}\left[ w_s = -\sum_{i=1}^{s-1} w_i \mid \pi(w) \not\in \V_q^{s-1} \right] \\
&= \frac{ 1 + (q-1) \left( 1 - \frac{ q\beta }{q-1} \right)^{s-1}}{q} \cdot (1 - \beta) + \left( 1 - \frac{ 1 + (q-1) \inparen{ 1 - \frac{ q\beta }{q-1} }^{s-1}}{q} \right) \cdot \frac{\beta}{q-1} \\
&= \frac{1}{q} + \inparen{ 1 - \frac{ q \beta }{q-1} }^s \inparen{ \frac{q-1}{q} }, 
\end{align*}
which establishes the inductive hypothesis for $s$.
\end{proof}

Motivated by the computations above, we can define the following useful shorthands:
\begin{definition} \label{def:Zpsi}
For $\lambda,\beta\in (0,\frac{q-1}q]$, define
\begin{equation}\label{eq:ZDef}
Z(\beta) = \Pr_{w\sim \mu_q^s}\left[w\in \V_q^s\right]=\frac{1+(q-1)\left(1-\frac{q\beta}{q-1}\right)^s}q,
\end{equation}
$$\psi(\lambda,\beta) = s\DKL{\lambda}{\beta}{q}+\log_q Z(\beta)$$
\end{definition}

From Equations (\ref{eq:P_lambda}), (\ref{eq:FuToFv}), (\ref{eq:P(wt(v)=lambda)}) and (\ref{eq:P(Fv=0)}), we conclude that
\begin{align}
\log_q P_\lambda &= t\log_q \Pr[F\bar u=0] \le tn \left(\DKL{\lambda}{\beta}{q}+\frac {\log_q\left(1+(q-1)\left(1-\frac{q\beta}{q-1}\right)^s\right)-1}{s}\right) \notag \\
&= (1-R)n\left(s\DKL{\lambda}{\beta}{q}+\log_q\left(1+(q-1)\left(1-\frac{q\beta}{q-1}\right)^s\right)-1\right) \notag \\
&= (1-R)n \psi( \lambda, \beta) \label{eq:lambda_bound}
\end{align}
for every $\beta\in \left(0,\frac{q-1}q\right]$.  Above, we have used the choice $t = (1-R)s$.  

This motivates the following definition:
\begin{definition}\label{def:phi}
Let $Z$ and $\psi$ be as in Definition~\ref{def:Zpsi}.  Define:
\begin{equation*} 
\varphi(\lambda) = \inf_{\beta\in (0,\frac{q-1}q]}\psi(\lambda,\beta).
\end{equation*}
\end{definition}

Definition~\ref{def:phi}, along with \eqref{eq:lambda_bound}, implies that
$\log_q P_\lambda \leq \varphi(\lambda)$, which establishes Item~\ref{item:phi_Bound} of Lemma~\ref{lem:phifacts}.
Next we establish
Item~\ref{item:P_lambdaBound} of Lemma~\ref{lem:phifacts}.  This follows from Definition~\ref{def:phi}, since
\[ \varphi(\lambda) \leq \psi(\lambda, \lambda) = \log_q\inparen{ 1 + (q-1) \inparen{ 1 - \frac{ q \lambda }{q-1} }^s } - 1, \]
using the fact that $\DKL{\lambda}{\lambda}{q} = 0$. 

This almost completes the proof of Lemma~\ref{lem:phifacts}, except for Item~\ref{item:P_lambdaMonotonicity}, which we establish in the next section using calculus.

\subsection{Proof of Item~\ref{item:P_lambdaMonotonicity} of Lemma~\ref{lem:phifacts}}\label{subsec:phi2}
In this section we prove Item~\ref{item:P_lambdaMonotonicity}, which will establish Lemma~\ref{lem:phifacts} and hence Theorem~\ref{thm:ldpc_distance}.
\begin{remark}[Difference between \cite{Gal62} and this proof]\label{rem:Gal-v-us}
This is the part of the proof where
the technical similarity between our proof and Gallager's breaks down.
The part of \cite{Gal62} which corresponds to our Item~\ref{item:P_lambdaMonotonicity} consists of an intricate analytic argument which does not seem (to us) to generalize to larger alphabets.
Thus, our proof has to rely on a different, more general, argument, which we give below.
\end{remark}
Before proving Item \ref{item:P_lambdaMonotonicity} of Lemma~\ref{lem:phifacts}, we need to better understand the relation between a given $\lambda\in (0,\frac{q-1}q]$, and the $\beta$ which minimizes the expression $\psi(\lambda,\beta)$.

\begin{lemma}\label{lem:lambda(beta)}
Let $\lambda\in (0,\frac{q-1}q]$. Then, $\psi(\lambda,\beta)$ is minimized by a unique $\beta\in (0,\frac {q-1}q]$. This $\beta$ is the only solution for
$$\EE_{w\sim \mu_q(\beta)}\left[\wt w\mid w\in \V_q^s\right] = \lambda.$$
\end{lemma}
\begin{proof}
We compute the derivative. 
\begin{align*}\label{eq:dLogZ}
\frac{d\log_e Z(\beta)}{d\beta} &= \frac1{\Pr_{w\sim \mu_q^s}[w\in \V_q^s]}\cdot \frac{d\left(\Pr_{w\sim \mu_q^s}[w\in \V_q^s]\right)}{d\beta} \\&= \frac1{\Pr_{w\sim \mu_q^s}[w\in \V_q^s]}\cdot \sum_{w\in \V_q^s}\frac{d \left(\frac{\beta}{q-1}\right)^{s\cdot\wt w}(1-\beta)^{s\cdot(1-\wt w)}}{d\beta}\\
&=\frac{\sum_{w\in \V_q^s}\left( \left(\frac{\beta}{q-1}\right)^{s\cdot \wt w}(1-\beta)^{s\cdot(1-\wt w)}\cdot s\cdot\left(\frac{\wt w}{\beta}-\frac{1-\wt w}{1-\beta} \right)\right)}{\Pr_{w\sim \mu_q^s}[w\in \V_q^s]}\\
&=s\cdot\left(\frac{\EE_{w\sim \mu_q^s}\left[\wt w\mid w\in \V_q^s\right]}{\beta} - \frac{1-\EE_{w\sim \mu_q^s}\left[\wt w\mid w\in \V_q^s\right]}{1-\beta}\right).\numberthis
\end{align*}
Also, it is not hard to see that
$$\frac{\partial\DKL{\lambda}{\beta}{q}}{\partial \beta} =  \log_q e\cdot\left(\frac{1-\lambda}{1-\beta}-\frac{\lambda}{\beta}\right).$$
Consequently,
\begin{align*}
\frac{\partial \psi(\lambda,\beta)}{\partial \beta} &=
s\frac{\partial\DKL{\lambda}{\beta}{q}}{\partial \beta} +\frac{d \log_q Z(\beta)}{d\beta}\\
&= \log_qe\cdot\left(\frac{s(1-\lambda)}{1-\beta}-\frac{s\lambda}{\beta}+\frac{d \log_e Z(\beta)}{d\beta}\right)\\
&=s\cdot\log_qe\cdot\left(\EE_{w\sim \mu_q^s}\left[\wt w\mid w\in \V_q^s\right]-\lambda\right)\left(\frac{1}{1-\beta}+\frac{1}\beta\right).
\end{align*}

We conclude that $\frac{\partial \psi(\lambda,\beta)}{\partial \beta}$ has the same sign as $\EE_{w\sim \mu_q^s}\left[\wt w\mid w\in \V_q^s\right]-\lambda s$. The lemma now follows from the following claim:
\begin{claim}\label{clm:lambdaIncreasing}
As $\beta$ increases in the range $(0, \frac{q-1}q]$ the function $\EE_{w\sim \mu_q^s}\left[\wt w\mid w\in \V_q^s\right]$ strictly increases from $0$ to $\frac{q-1}q$.
\end{claim}
\begin{proof}
Due to (\ref{eq:ZDef}) and (\ref{eq:dLogZ}),
\begin{align*}
\EE_{w\sim \mu_q^s}\left[\wt w\mid w\in \V_q^s\right] &= \left(\frac{d \log_e Z(\beta)}{s\cdot d\beta} +\frac{1}{1-\beta}\right)\beta(1-\beta)\\
&= \left(\frac{\frac{d Z(\beta)}{d\beta}}{s\cdot Z(\beta)} +\frac{1}{1-\beta}\right)\beta(1-\beta)\\
&=\left(\frac{-q\left(1-\frac{q\beta}{q-1}\right)^{s-1}}{1+(q-1)\left(1-\frac{q\beta}{q-1}\right)^{s}} +\frac{1}{1-\beta}\right)\beta(1-\beta)\\
&= \beta \frac{1-\left(1-\frac{q\beta}{q-1}\right)^{s-1}\cdot(1+q\beta)}{1+(q-1)\left(1-\frac{q\beta}{q-1}\right)^{s}},\numberthis\label{eq:E(winV)}
\end{align*} 
and the claim readily follows.
\end{proof}
The proof of the lemma is thus concluded. 
\end{proof}

Lemma \ref{lem:lambda(beta)} and Claim \ref{clm:lambdaIncreasing} justify the following definition:
\begin{definition}
For $\lambda\in (0,\frac {q-1}q]$, denote the $\beta\in(0,\frac {q-1}q]$ which minimizes $\psi(\lambda,\beta)$ by $\beta(\lambda)$. The inverse of this function is denoted $\lambda(\beta)$. 
\end{definition}
By Lemma \ref{lem:lambda(beta)} and Equation (\ref{eq:E(winV)}), 
\begin{equation}\label{eq:lambda(beta)}
\lambda(\beta) =  \beta \frac{1-\left(1-\frac{q\beta}{q-1}\right)^{s-1}}{1+(q-1)\left(1-\frac{q\beta}{q-1}\right)^{s}}.
\end{equation}
\begin{remark}
Unfortunately, there are good reasons to suspect that the function $\beta(\lambda)$ has no closed-form expression (see, e.g., the discussion about backward mapping in \cite[Sec. 3.4.2]{WJ08}), so we prefer to work with its inverse.
\end{remark}

It is convenient to extend the definition of these functions to the closed interval $[0,\frac {q-1}q]$ by taking limits, namely, $\lambda(0) = \beta(0) = 0$, and 
\begin{align*}
\varphi(0) &= \lim_{\lambda\to 0}\varphi(\lambda) = \lim_{\lambda\to 0} \psi(\lambda,\beta(\lambda)) \lim_{\beta\to 0} \psi(\lambda(\beta),\beta) = \lim_{\beta\to 0}\DKL{\lambda(\beta)}{\beta}{q}+\log_qZ(\beta) \\
&= \lim_{\beta\to 0}\DKL{\lambda(\beta)}{\beta}{q} = \lim_{\beta\to 0}-\lambda(\beta)\log_q \beta = 0.
\end{align*}

We are now able to prove Item \ref{item:P_lambdaMonotonicity} of Lemma~\ref{lem:phifacts}.
\begin{proof}[Proof of Lemma~\ref{lem:phifacts}, Item~\ref{item:P_lambdaMonotonicity}]
Let $\alpha(\lambda) = \frac{\varphi(\lambda)}{h_q(\lambda)}$. The claim follows immediately from the four following claims:
\begin{claim}\label{clm:alphaRightBound}
$\alpha(\frac{q-1}q) = -1$.
\end{claim}
\begin{claim}\label{clm:alphaNearRightBound}
$\alpha(\lambda) < -1$ for some $\lambda\in (0,\frac {q-1}q)$.
\end{claim}
\begin{claim}\label{clm:alphaLowerBoundNear0}
There exists $\epsilon > 0$ such that $\alpha(\lambda)> -\frac s2$ for all $\lambda\in (0,\epsilon)$.
\end{claim}
\begin{claim}\label{clm:alphaSols}
For each $y\in(-\frac s2,-1]$, the equation $\alpha(\lambda)= y$ has at most one solution $\lambda\in (0,\frac {q-1}q]$. 
\end{claim}	

Indeed, Claims \ref{clm:alphaRightBound} and \ref{clm:alphaSols} show that $\alpha(\lambda)\ne -1$ for $\lambda < \frac {q-1}q$. Since $\alpha$ is continuous, it is either upper bounded or lower bounded by $-1$ in the whole range $(0,\frac{q-1}q].$ Claim \ref{clm:alphaNearRightBound} implies the former. By Claim \ref{clm:alphaSols}, if $-\frac s2<\alpha(\lambda_0) < -1$ for some $\lambda_0 \in (0,\frac{q-1}q)$, then $\alpha$ must be strictly increasing in the range $[\lambda_0,\frac {q-1}q]$. The lemma now follows from Claim \ref{clm:alphaLowerBoundNear0}. We proceed to prove these claims.

\begin{proof}[Proof of Claim \ref{clm:alphaRightBound}]
Note that $\alpha(\frac{q-1}q) = \varphi(\frac{q-1}q)$.
Due to Item~\ref{item:phi_Bound}, $$\varphi\left(\frac{q-1}q\right) \le -1.$$ In the reverse direction, 
\begin{align*}
\varphi(\lambda) = \min_\beta \psi(\lambda,\beta) &= \min_\beta \left(s\cdot \DKL{\lambda}{\beta}{q}+\log_q Z(\beta)\right) \\ &\ge \min_\beta \left(s\cdot \DKL{\lambda}{\beta}{q}\right)-1\ge -1
\end{align*}
for all $\lambda$. The first inequality above holds since $Z(\beta) \ge \frac 1q$, due to (\ref{eq:ZDef}) .
\end{proof}

\begin{proof}[Proof of Claim \ref{clm:alphaNearRightBound}]
By Item \ref{item:P_lambdaBound}, 
\begin{equation}\label{eq:alpha(lambda)Bound}
\alpha(\lambda) \le \frac{\log_q\left(1+(q-1)\left(1-\frac{q}{q-1}\lambda\right)^s\right)-1}{h_q(\lambda)}.\end{equation}
Let $\lambda= \frac{q-1}q-\epsilon$. As $\epsilon$ tends from above to $0$, the numerator of (\ref{eq:alpha(lambda)Bound})'s right-hand side is $-1+\Theta(\epsilon^s)$, while the denominator is $1-\Theta(\epsilon^2)$. Thus, for $\epsilon$ small enough, (\ref{eq:alpha(lambda)Bound}) yields $\alpha(\lambda)<-1$.
\end{proof}

\begin{proof}[Proof of Claim \ref{clm:alphaLowerBoundNear0}]
Let $$\bar Z(\beta) = \Pr_{w\sim \B_q^s}\left(w\in \V_q^s\wedge \wt w\le \frac 2s\right) = (1-\beta)^s+\binom s2 (1-\beta)^{s-2}\beta^2$$
and $$\bar \psi(\beta,\lambda) = s\DKL{\lambda}{\beta}{q}+\log_q\bar Z(\beta).$$
Clearly, $\bar{\psi}(\beta,\lambda)$ is a lower bound on $\psi(\beta,\lambda)$, so
$$\varphi(\lambda) \ge \min_{\beta\in (0,\frac{q-1}q]}\bar\psi(\lambda,\beta).$$
Note that
$$\frac{\partial\bar\psi(\lambda,\beta)}{\partial \beta} = \frac{s}{\beta(1-\beta)}\left(\frac{2(s-1)}{\left(\frac{1-\beta}{\beta}\right)^2+\binom s2}-\lambda\right),$$
Hence, for $\lambda< \frac2s$, the minimum of $\bar\psi(\lambda,\beta)$ is attained at $\beta_0 = \frac{y}{1+y}$, where $$y=\left(\frac{\lambda}{2(s-1)-\binom s2\lambda}\right)^{\frac 12}.$$
Therefore, 
$$\alpha(\lambda)= \frac{\varphi(\lambda)}{h_q(\lambda)}\ge \frac{\bar\psi(\lambda,\beta_0)}{h_q(\lambda)} = \frac s2\left(-1+\frac{\lambda\left(\log_q\left(2(s-1)-\binom s2\lambda\right)-\log_q(1-\lambda s)\right)+(1-\lambda)\log_q(1-\lambda)}{h_q(\lambda)}\right).$$
For $\lambda$ small enough, the right-hand side is clearly larger than $-\frac s2$.
\end{proof}
\begin{proof}[Proof of Claim \ref{clm:alphaSols}]
Denote $\beta^* = \beta(\lambda)$. Let $y\in (-\frac s2,-1]$, and define the function $\varphi_y(\lambda) = \varphi(\lambda)-yh_q(\lambda)$. We seek to show that $\varphi_y(\lambda)$ has at most one root in the range $(0,\frac{q-1}q]$. This is a consequence of the following three statements, proven below:
\begin{enumerate}
	\item \label{itm:dphiyOneExtremum} $\frac{d\varphi_y(\lambda)}{d\lambda}$ has at most one extremal point in the open interval $(0,\frac{q-1}q)$.
	\item \label{itm:dphiyRoot} $\frac{d\varphi_y(\lambda)}{d\lambda}(\frac{q-1}q) = 0$.
	\item \label{itm:phiyRoot}$\varphi_y(0)=0$.
\end{enumerate}
Indeed, the first statement implies that $\frac{d\varphi_y(\lambda)}{d\lambda}$ has at most two roots in the interval $(0,\frac{q-1}q]$. The second statement says that one of these roots is at $\frac{q-1}q$, so $\frac{d\varphi_y(\lambda)}{d\lambda}$ has at most one root in $(0,\frac{q-1}q)$. Consequently $\varphi_y(\lambda)$ has at most one extremal point and two roots in $[0,\frac{q-1}q]$. Due to the third statement, one of these roots is $0$, so there can only be one root in $(0,\frac{q-1}q]$. We turn to prove these statements.

Statement \ref{itm:phiyRoot} is trivial. For Statement \ref{itm:dphiyRoot}, note that in the derivative
$$\frac{d\varphi(\lambda)}{d\lambda} = \frac{\partial\psi(\lambda,\beta)}{\partial{\beta}}_{\beta=\beta^*}\cdot\frac{d\beta^*}{d\lambda}+\frac{\partial \psi(\lambda,\beta)}{\partial\lambda}_{\beta=\beta^*},$$
the first term vanishes since $\psi$ has a minimum at $(\lambda, \beta^*)$. Hence,
$$\frac{d\varphi(\lambda)}{d\lambda} = \frac{\partial \psi(\lambda,\beta)}{\partial\lambda}_{\beta=\beta^*} = s\frac{\partial \DKL{\lambda}{\beta}{q}}{\partial \lambda}_{\beta = \beta^*} = s\log_q\frac{\lambda(1-\beta^*)}{(1-\lambda)\beta^*}.$$
In particular, $\beta(\frac{q-1}q) = \frac{q-1}q$, so 
$$\frac{d\varphi_y(\lambda)}{d\lambda}_{\lambda=\frac{q-1}q} = \frac{d\varphi(\lambda)}{d\lambda}_{\lambda=\frac{q-1}q} - y\frac{dh_q(\lambda)}{d\lambda}_{\lambda=\frac{q-1}q}= 0,$$
since, in the last transition, the two terms vanish.

We turn to Statement \ref{itm:dphiyOneExtremum}. Define the new variable $x=1-\frac {q\beta^*}{q-1}$. Note the following useful relations, the second of which follows from Equation (\ref{eq:lambda(beta)}):
\begin{equation}\label{eq:betaInX}
\beta^* = \frac {q-1}q(1-x)
\end{equation}
and
\begin{equation}\label{eq:lambda/1-lambdaInX}
\frac{\lambda}{1-\lambda} = \frac{\beta^*}{1-\beta^*}\cdot \frac{1-x^{s-1}}{1+(q-1)x^{s-1}}.
\end{equation}
By (\ref{eq:betaInX}) and (\ref{eq:lambda/1-lambdaInX}),
\begin{align*}\frac{d \varphi_y(\lambda)}{d\lambda} &=
s\frac{\partial \DKL{\lambda}{\beta}{q}}{\partial \lambda}_{|\beta = \beta^*}-y\frac{dh_q(\lambda)}{d\lambda} \\ &=  s\log_q\frac{\lambda(1-\beta^*)}{(1-\lambda)\beta^*} + y\log_q\frac{\lambda}{1-\lambda} \\ &= 
 s\log_q\frac{1-\beta^*}{\beta^*} + (s+y)\log_q\frac{\lambda}{1-\lambda}\\
&=-y\log_q\frac{1+(q-1)x}{(q-1)(1-x)} + (s+y)\log_q\frac{1-x^{s-1}}{1+(q-1)x^{s-1}}.
\end{align*}
Now,
$$
\frac{d^2\varphi_y(\lambda)}{dx d\lambda}\cdot\ln q = \frac{-yq}{(1+(q-1)x)(1-x)} - \frac{(s+y)(s-1)qx^{s-2}}{\left(1-x^{s-1}\right)\left(1+(q-1)x^{s-1}\right)}.
$$
This second derivative vanishes when
\begin{align*}
\frac{-(s+y)}{y}&=\frac{\left(1-x^{s-1}\right)\left(1+(q-1)x^{s-1}\right)}{(s-1)(1+(q-1)x)(1-x)x^{s-2}}.
\end{align*}
Equivalently,
\begin{equation} \label{eq:SecondDerivativeCondition}
\frac{-(s+y)}{y} = \frac 1{s-1}\sum_{i=0}^{s-2}\frac{x^{-i}+(q-1)x^{i+1}}{1+(q-1)x}.
\end{equation}
By examining each term of this sum separately, it is straightforward to verify that the right-hand side of (\ref{eq:SecondDerivativeCondition}) is a convex function of $x$, which tends to $\infty$ (resp. $1$) as $x\to 0$ (resp. $x\to 1$). Since $y>-\frac s2$, the left-hand side of (\ref{eq:SecondDerivativeCondition}) is larger than $1$, so there is a unique $x\in(0,1)$ which solves (\ref{eq:SecondDerivativeCondition}). Statement \ref{itm:dphiyOneExtremum} follows.
\end{proof}
This establishes Item~\ref{item:P_lambdaMonotonicity} of Lemma~\ref{lem:phifacts}.
\end{proof}

Having completed the proof of Lemma~\ref{lem:phifacts}, we have finished the proof of Theorem~\ref{thm:ldpc_distance}.